\documentclass[11pt]{article}
\usepackage[utf8]{inputenc}

\usepackage{amssymb, amsmath, amsthm, graphicx, authblk, bm, bbm, fullpage, soul}
\usepackage[round]{natbib}
\usepackage{enumitem}
\usepackage{subcaption}
\usepackage{float}
\usepackage{comment}
\usepackage{xcolor}
\usepackage{booktabs}

\usepackage{hyperref}
\hypersetup{
    colorlinks=true,
    linkcolor=blue,
    filecolor=magenta,
    urlcolor=blue,
    citecolor=blue,
}
\usepackage{cleveref}

\usepackage{algorithm}
\usepackage{algorithmic}

\DeclareMathOperator*{\argmin}{arg\,min}

\DeclareMathOperator*{\lag}{lag}
\DeclareMathOperator*{\train}{train}

\DeclareMathOperator*{\diag}{diag}

\DeclareMathOperator*{\corr}{Corr}

\newcommand{\ri}{{\rho}}

\newcommand{\F}{{\mathrm F}}

\DeclareMathOperator*{\rank}{Rank}

\DeclareMathOperator*{\op}{op}

\DeclareMathOperator*{\svd}{SVD}
\DeclareMathOperator*{\rsvd}{RSVD}

\newcommand{\p}{\mathbb P}

\newcommand{\ii}{{(i)}}

\newcommand{\dint}{ \mathrm{d}}
\newcommand{\bb}{{\mathfrak b}}

\newcommand{\lt}{{\mathbf L_2}}
\newcommand{\E}{\mathbb{E}}

\newtheorem{lemma}{Lemma}
\newtheorem{theorem}{Theorem}
\newtheorem{definition}{Definition}
\newtheorem{remark}{Remark}

\title{Online Change Point  Detection for Multivariate Inhomogeneous  Poisson Processes Time Series}

\begin{document}

\author[1]{Xiaokai Luo}
\author[2]{Haotian Xu}
\author[3]{Carlos Misael Madrid Padilla}
\author[4]{Oscar Hernan Madrid Padilla}
\affil[1]{Department of Applied and Computational Mathematics and Statistics, University of Notre Dame}
\affil[2]{Department of Mathematics and Statistics, Auburn University}
\affil[3]{Department of Statistics and Data Science,  Washington University in Saint Louis}
\affil[4]{Department of Statistics, University of California, Los Angeles}

\maketitle
\begin{abstract}
  We study online change point detection for multivariate inhomogeneous Poisson point process time series. This setting arises commonly in applications such as earthquake seismology, climate monitoring, and epidemic surveillance, yet  remains underexplored in the machine learning and statistics literature.   We propose a method that  uses   low-rank matrices  to represent the  multivariate Poisson intensity functions, resulting in an adaptive nonparametric detection procedure. Our algorithm is single-pass and requires only constant computational cost per new observation, independent of the elapsed length of the time series.  We  provide theoretical guarantees to control  the overall false alarm probability and  characterize  the detection delay under temporal dependence. We also develop a new Matrix Bernstein inequality for temporally dependent Poisson point process time series, which may be of independent interest.
  Numerical experiments demonstrate that our method is both statistically  robust and computationally efficient.
\end{abstract}

\section{Introduction}
An inhomogeneous Poisson point process (PPP) provides a flexible model for random events  occurring in space with location-dependent rates. 
Applications include forest fires \citep{stoyan2000recent}, earthquakes \citep{bray2013assessment}, citywide crime incidents \citep{baddeley2021analysing}, and epidemic outbreaks \citep{AlDousari2021NHPPKuwait}. In these examples, the \emph{intensity function} is central, as it specifies the expected event density and encodes the spatial structure of the phenomenon.
Formally, a PPP sample  $X\subset \mathbb R^d$ is said to be sampled from an intensity $\lambda:\mathbb R^d\to\mathbb R_+$ if:

\noindent\textbf{(1)} for every   set $S\subseteq\mathbb R^d$, the count $|S\cap X|$ is a Poisson random variable  with mean $\int_S \lambda(x)\,dx$;   and

\noindent\textbf{(2)} for disjoint   sets $S_1,\ldots,S_n\subset\mathbb R^d$, the random variables $|S_1\cap X|,\ldots,|S_n\cap X|$ are independent.

Estimation for a \emph{single} PPP is well studied: \citet{reynaud2003adaptive} derive minimax rates of estimating the intensity function in one dimension; \citet{flaxman2017poisson} estimate intensities using reproducing kernel Hilbert space (RKHS); and \citet{xu2025supp} leverage tensor structure for multivariate intensity estimation.

In this work, we study online change point detection for PPP time series. Specifically, we assume that at every instance $i$, we observe a PPP sample $ X^{(i)}$, and the underlying marginal intensity function can changes over time in an abrupt manner.   A key feature of real data is that consecutive samples  $X^{(i)}$ are temporally  dependent. For example, environmental persistence   and latent dynamics can induce temporal dependence across time points even when each sample is well approximated by a PPP marginally \citep{baddeley2007spatial}. To model  this feature, we allow the sequence $\{X^{(i)}\}_{i= 1}^\infty $ to exhibit temporal dependence, while maintaining PPP structure within each time index $i$. We formalize this  general temporal  dependent time series models below.
\begin{definition}[PPP time series   with temporal dependence]
\label{def:model_ppp_beta_mixing}
Let $\mathbb X\subset\mathbb R^d$ be compact. Let $\{X^{(i)}\}_{i=1}^\infty $ be a
 sequence of point processes on $\mathbb X$.
Assume the following two conditions hold.

\noindent {\bf (A) Poisson marginals.}  
For each time index $i$, the marginal distribution of $X^{(i)}$ corresponds to an inhomogeneous Poisson point process on $\mathbb{X}$ with   intensity $\lambda_i^*:\mathbb{X}\to\mathbb{R}_+$, satisfying
\begin{equation*}
\sup_{i\in\mathbb Z}\|\lambda_i^*\|_{\infty}<\infty \quad \text{and} \quad \sup_{i\in\mathbb Z}\|\lambda_i^*\|_{W^{ 2, \gamma  }}<\infty,
\end{equation*}
where $\|\cdot\|_{W^{2,\gamma}}$ is the Sobolev norm defined in \eqref{eq:W-norm}.

\noindent {\bf (B) Geometric $\beta$-mixing across time.}  
There exists a constant $c>0$ such that
\[
\beta(k)\le e^{-c(k-1)}\qquad\text{for all }k =  1,2,3,\ldots.
\]
Here, $\{\beta(k)\}_{k= 1}^\infty $ are the $\beta$-mixing coefficients of $\{X^{(i)}\}_{i=1}^\infty$, as defined in
\eqref{def:beta-mixing}.

\noindent Under {\bf (A)} and {\bf(B)}, we consider the  possible two scenarios.

\noindent {\bf (M0) No change point.}  
There exists an intensity $\lambda^*:\mathbb X\to\mathbb R_+$ such that
\[
\lambda_i^*=\lambda^* \text{ for all } i\in\mathbb Z.
\]
\noindent {\bf (M1) Single change point.}  
There exist intensities $\lambda^*,\lambda_a^*:\mathbb X\to\mathbb R_+$ with $\lambda^*\neq \lambda_a^*$ and
a change point $\bb\in\mathbb Z_+$ such that
\[
\lambda_i^*=\lambda^* \text{ for all } i\le \bb,
\quad \text{and} \quad 
\lambda_i^*=\lambda_a^* \text{ for all } i> \bb.
\]
\end{definition}
While real-world time series may contain multiple change points, we follow the online change point detection literature and focus on the above at most one change point settings. In practice, the detection algorithm is simply restarted immediately after a change is declared.

\subsection{Related Works}
Nonparametric change detection for general distributions has been studied through a range of modern approaches. Representative lines of work include kernel methods that embed distributions into reproducing kernel Hilbert spaces and enable sequential two-sample testing \citep{harchaoui2008kernel,li2015m,arlot2019kernel,wei2022online}, discrepancy measures based on energy distances and characteristic functions \citep{matteson2014nonparametric}, and learning-based procedures such as neural network detectors \citep{li2024automatic,gong2022neural}. Additional perspectives include functional kernel approaches \citep{romano2023fast}, random forest methods \citep{londschien2023random}, and graph-based two-sample statistics \citep{chen2015graph,chu2019asymptotic}. Anytime-valid methodology via e-values has also been emphasized recently, providing online change detectors with rigorous error control under minimal assumptions \citep{shin2022detectors}.
From a theoretical standpoint, minimax-optimal results for offline nonparametric change point detection and localization have been established for changes in smooth distributions \citep{padilla2021optimal,madrid2023change}.

Change point detection for point process time series has wide-ranging applications, including earthquake seismology \citep{ogata2011significant}, wildfire monitoring \citep{xu2011point}, epidemic surveillance \citep{hohl2020daily} and DNA sequencing \citep{zhang2016scan}. 
{To the best of our knowledge, existing point process change point methods are largely tailored to parametric temporal or spatio-temporal models, or to offline estimation, and therefore are not directly applicable to our online nonparametric multivariate PPP-intensity setting.
In particular, \citet{wang2023sequential} consider sequential detection for self- and mutually-exciting point processes (specifically, Hawkes networks) using a parametric CUSUM/likelihood-based construction on temporal event data. Similarly, \citet{zhang2023online} study online score statistics for clustered changes in multivariate Hawkes network point processes. The composite-likelihood approach of \citet{zhao2019composite} is also model-based and focuses on offline change-point estimation in piecewise stationary spatio-temporal processes. Finally, \citet{dion2023multiple} study multiple \emph{offline} change-point detection for some point processes, including inhomogeneous and marked Poisson processes, using a minimum-contrast estimator.  Complementing algorithmic developments, \citet{brandenberger2025detecting} study fundamental detection limits for point-process changes from an information-theoretic perspective.}


Despite these advances, existing change point methods for point processes often rely on strong parametric assumptions and typically assume independence across time. Our PPP time series setting poses additional challenges: each observation is an unordered random set with random cardinality, temporal dependence is present over time, and the signal corresponds to a change in an intensity function rather than a finite-dimensional parameter. 
To our knowledge, there is no general-purpose approach with provable guarantees for detecting nonparametric changes in multivariate PPP intensities under realistic temporal dependence.

\subsection{Summary of Results}

\noindent\textbf{New algorithm for online change detection.} We introduce a new computationally efficient online nonparametric detection procedure for PPP time series. The key idea is to map each PPP sample $X^{(i)}$  to a low-rank intensity matrix yielding a scalable algorithm adaptive to local changes in the underlying intensity. In particular,
the per new observation cost is constant independent of the elapsed time series length. As a result, our method is single-pass, and the total computational cost scales linearly with the number of observed samples.

\noindent\textbf{General theory.}
We develop a  theoretical  framework for online change detection in multidimensional PPPs under temporal dependence. A key technical ingredient is a new Matrix Bernstein inequality for geometrically $\beta$-mixing PPP time series (\Cref{thm:bernstein_beta_mixing_ppp_rect_nonstat} in Appendix). Combined with our low-rank intensity matrix approximation analysis, this yields sharp non-asymptotic bounds that simultaneously account for basis-truncation bias and stochastic variance in dependent PPP time series.

\noindent\textbf{Finite-sample guarantees.}
We establish finite-sample guarantees to show that our newly proposed method both controls the  overall false alarm probability and can detect a true change within a delay that depends explicitly on the change size in the intensity functions.  
 
\noindent\textbf{Empirical evidence.}
We demonstrate that our  procedure reliably identifies meaningful intensity changes while remaining computationally efficient through extensive simulations and a real-data application in modeling earthquake activity.

\subsection{Notations}

\noindent\textbf{Matrices.}
For $\mathcal M\in\mathbb R^{m\times n}$, let $\mathcal M_{(\mu,\eta)}$ be its $(\mu,\eta)$ entry.
We write $\|\mathcal M\|_{\F}$ and $\|\mathcal M\|_{\op}$ for the Frobenius and operator norms, and
$\rank(\mathcal M)$ for the rank. If $\rank(\mathcal M)=s$, let
$\sigma_1(\mathcal M)\ge \cdots \ge \sigma_s(\mathcal M)>0$ be its nonzero singular values.
With the SVD $\mathcal M=U\Sigma V^\top$, define the rank-$r$ truncation ($r\le s$) by
$$\svd(\mathcal M,r)=U\,\Sigma_{(r)}\,V^\top.$$ Here $\Sigma_{(r)}=\diag(\sigma_1(\mathcal M),\ldots,\sigma_r(\mathcal M),0,\ldots,0)\in\mathbb R^{s\times s}$.

\noindent\textbf{Function spaces.}
Let $\mathbb X=\mathbb X_1\times\cdots\times\mathbb X_d\subset\mathbb R^d$ be compact and define
 $$
\lt(\mathbb X)=\big\{f:\mathbb X\to\mathbb R:\ \|f\|_{\lt}^2=\int_{\mathbb X}f^2(x)\,\dint x<\infty\big\}.
$$
For convenience we often take $\mathbb X_1=\cdots=\mathbb X_d=\Omega$.
A family $\{\phi_k\}_{k\ge1}\subset \lt(\Omega)$ is orthonormal if
$$\int_\Omega \phi_k(x)\phi_j(x)\,\dint x=\mathbf 1\{k=j\}.$$
For an integer $\gamma\ge1$, let $W^{2,\gamma}(\mathbb X)$ be the Sobolev space of functions with weak derivatives
$f^{(b)}\in\lt(\mathbb X)$ for all multi-indices $b$ with $|b|\le\gamma$, equipped with the norm
\begin{equation}\label{eq:W-norm}
\|f\|_{W^{2,\gamma}}^2=\sum_{|b|\le\gamma}\|f^{(b)}\|_{\lt}^2.
\end{equation}

\noindent\textbf{Temporal dependence ($\beta$-mixing).}
Let $\{X_i\}_{i=-\infty}^\infty$ be a time series and define
$\mathcal F_{-\infty}^{\,j}=\sigma(X_i:i\le j)$ and
$\mathcal F_{j+k}^{\,\infty}=\sigma(X_i:i\ge j+k)$
for $j\in\mathbb Z$ and $k\ge1$. The $\beta$-mixing coefficients are
\begin{equation}\label{def:beta-mixing}
\beta(k)=\sup_{j\in\mathbb Z}\beta\big(\mathcal F_{-\infty}^{\,j}, \mathcal F_{j+k}^{\,\infty}\big),
\qquad k\ge1,
\end{equation}
where, for two $\sigma$-fields $\mathcal A$ and $\mathcal B$,
the absolute regularity coefficient is defined by
\begin{align*}
    \beta(\mathcal A,\mathcal B)
=\frac12\sup\bigg\{\sum_{i=1}^I\sum_{j=1}^J
|\mathbb P(A_i\cap B_j)-\mathbb P(A_i)\mathbb P(B_j)|:
\\
\{A_i\}_{i=1}^I\subset\mathcal A, \{B_j\}_{j=1}^J\subset\mathcal B\ \text{are   partitions}\bigg\}.
\end{align*} 
 We refer interested readers to \citet{bradley2005basic} for a detailed discussion of mixing conditions. 
\section{Online Change Point Detection for Poisson Point Processes}
We now describe the time series  model and our detection procedure.  We work under the   PPP  time series models with temporal dependence
introduced in Definition~\ref{def:model_ppp_beta_mixing}.

\noindent\textbf{Training stage.} We observe a collection of point processes
\[
X^{(i)} \subset \mathbb X = \Omega^{\otimes d} \subset \mathbb R^d,\quad i=1,\ldots,N_{\train},
\]
where each $X^{(i)}$ is generated over a fixed observation window (e.g., one day of events, one spatial
snapshot, or one short space--time interval).
We focus on the multivariate settings with $d\ge 2$. The case   $d=1$ is discussed in \Cref{remark:1d poisson}.  We assume $\{X^{(i)}\}_{i = 1}^{N_{\train}}$ to follow
\textbf{(M0)} in Definition~\ref{def:model_ppp_beta_mixing} with pre-change intensity $\lambda^*$.
The training sample is used to calibrate tuning parameters and the detection threshold.

\noindent\textbf{Post-training stage.} After the training stage, the sequence evolves according to one of two scenarios. 
In the first scenario \textbf{(M0)}, the Poisson intensity remains unchanged for the rest of the time horizon; that is,
$
\{X^{(i)}\}_{i=1}^{\infty}
$ is   stationary with common intensity function $\lambda^{*}$.
In the second scenario \textbf{(M1)}, there exists an unknown change point $\bb\ge N_{\train}$ such that
$  
\{ X^{(i)}\}_{i=N_{\train} + 1}^{\bb}
$  is  stationary with intensity $\lambda^*$ and $  
\{ X^{(i)}\}_{i=\bb+1}^{\infty}
$  is  stationary with intensity $\lambda^*_a$.
We allow temporal dependence within each regime under a geometric $\beta$-mixing condition, as specified in Definition~\ref{def:model_ppp_beta_mixing} {\bf (A)}. 

\noindent\textbf{Goal.} Given the training and post-training data, our goal is to develop an online algorithm such that (i) when there is  no change, the overall probability of a false alarm is kept small; and (ii) if a change occurs after the training phase, the algorithm raises an alarm as quickly as possible to  minimize the detection delay.

\noindent\textbf{Mapping intensity functions to matrices.}
We represent  the  infinite-dimensional  intensity functions to a  finite-dimensional  matrix via a RKHS basis.  This representation    has a distance-preserving property, enabling efficient operation over arbitrary intervals via dynamic programming, and avoids additional Monte Carlo procedures to compute $\lt$ norms of functions   in higher dimensions.

Formally, let $x=(x_1,\ldots,x_d)\in \mathbb{X}=\Omega^{\otimes d}\subset \mathbb{R}^d$, and partition the index set $[d] =\{1,\ldots,d\}$ into two disjoint subsets $\mathcal{I}_1$ and $\mathcal{I}_2$ such that
\[
[d]=\mathcal{I}_1   \cup  \mathcal{I}_2,
\qquad |\mathcal{I}_1|=p,\ \ |\mathcal{I}_2|=q,\ \ p+q=d .
\]
Define the corresponding coordinate partition as
\[
y  = (x_j)_{j\in \mathcal{I}_1}\in \Omega^{\otimes p},
\qquad
z  = (x_j)_{j\in \mathcal{I}_2}\in \Omega^{\otimes q},
\]
we write $x=(y,z)\in \Omega^{\otimes p}\times \Omega^{\otimes q}\subset \mathbb{R}^{p+q}.$

Let $\{\phi_k\}_{k =1}^\infty $ form  orthonormal univariate basis of $\lt(\Omega)$.
Then 
$  \big \{\phi_{i_1}(x_1)\cdots \phi_{i_p}(x_p)\big \}_{i_1,\ldots,i_p=1}^{\infty }      $
is a set of complete basis functions of $ \lt(\Omega^{\otimes p})$.
 For any $M \in \mathbb Z_+$,   the collection of functions $    \{\phi_{i_1}(x_1)\cdots \phi_{i_p}(x_p)\big \}_{i_1,\ldots,i_p=1}^{M}     \subset \lt (  \Omega^{\otimes p} ) $ 
 is orthonormal in $\Omega^{\otimes p}$ with cardinality $M^p$.
 Ordering the multi-indices $(i_1,\ldots, i_p)$  and $(\ell_1,\ldots, \ell_q)$,  we denote 
\begin{equation}
\label{eq:basis-blocks}
\begin{aligned}
\big\{\Phi_\mu(y)\big\}_{\mu=1}^{M^p}
&= \Big\{\phi_{i_1}(x_1)\cdots \phi_{i_p}(x_p)\Big\}_{i_1,\ldots,i_p=1}^{M},\\
\big\{\Psi_\eta(z)\big\}_{\eta=1}^{M^q}
&= \Big\{\phi_{\ell_1}(x_{p+1})\cdots \phi_{\ell_q}(x_{p+q})\Big\}_{\ell_1,\ldots,\ell_q=1}^{M}.
\end{aligned}
\end{equation} 
Let $\lambda^*: \Omega^{\otimes p} \times \Omega^{\otimes q} \to\mathbb R_+ $ satisfy $\|\lambda^*\|_{W^{ 2, \gamma  }}<\infty$. Define the matrix $\mathcal M(\lambda^*)\in\mathbb R^{M^p\times M^q}$ by
\begin{equation}
\label{eq:define M star}
\mathcal M(\lambda^*)_{(\mu,\eta)}
= \iint_{\mathbb R^{p+q}} \lambda^*(y,z)\,\Phi_\mu(y)\,\Psi_\eta(z)\dint y\dint z.
\end{equation}
Using $\mathcal M(\lambda^*)$, we can approximate $\lambda^*$ by its truncated expansion:
\begin{equation}
\label{eq:density finite matrix main}
\lambda_M^*(y,z)
= \sum_{\mu=1}^{M^p}\sum_{\eta=1}^{M^q} \mathcal M(\lambda^*)_{\mu,\eta} \,\Phi_\mu(y)\,\Psi_\eta(z).
\end{equation}
It was shown in Appendix G.1  that if $\{\phi_k\}_{k=1}^\infty $ are univariate RKHS orthonormal basis function, then
\begin{equation}  \label{eq:density finite matrix error main}
\|\lambda^*-\lambda_M^*\|_{\lt} \le C\,\|\lambda^*\|_{W^{ 2,\gamma }}\,M^{-\gamma}.
\end{equation}

\begin{remark}[Coordinate split]
It follows from \eqref{eq:density finite matrix main} and \eqref{eq:density finite matrix error main} that
$\mathcal M(\lambda^*)$ provides an accurate matrix representation of $\lambda^*$ with small approximation error.
This representation requires a coordinate partition in $\mathbb X \subset \mathbb R^d$. 
As suggested by \Cref{theorem:svd Hilbert} in Appendix, such a split is  valid for any function that admits a functional PCA representation. 
The split can also be  specified using prior
knowledge of the dataset, as demonstrated  in our real-data example. As a third option, one can partition the features into two
groups so that variables are more correlated within groups and less correlated across groups. See \Cref{remark:tuning parameters} for more details. 
\end{remark}

\subsection{Online change point detection}
For each process $X^{(i)}$, define its intensity matrix by $\widehat{\mathcal M}^{(i)}\in\mathbb R^{M^p\times M^q}$ with entries
\begin{equation}
\label{eq:A-i}
\widehat{\mathcal M}^{(i)}_{(\mu,\eta)}
= \sum_{x^{(i)}=(y^{(i)},z^{(i)})\in X^{(i)}} \Phi_\mu(y^{(i)})\,\Psi_\eta(z^{(i)}).
\end{equation}
It follows from Campbell’s Theorem (\Cref{lemma:campbell} in Appendix) that 
$ \E(\widehat{\mathcal M}^{(i)}_{(\mu,\eta) }) =\mathcal M(\lambda^*)_{  (\mu,\eta)}  $. Therefore, under the single change point scenario \textbf{(M1)}, the intensity functions admit   a change at \(\bb\), and
\begin{align} \label{eq:expected value of matrix at time t}
    \E\big(\widehat{\mathcal M}^{(i)}\big)= \begin{cases}
        \mathcal M(\lambda^*) &\text{if } i\le \bb,
        \\
        \mathcal M(\lambda^*_a) &\text{if } i> \bb.
    \end{cases}
\end{align}
  Consequently, for any $n\le \bb$, the deviation
 between  the matrices $ n^{-1}\sum_{i=1}^{n}\widehat{\mathcal M}^{(i)} $ and $  \mathcal M(\lambda^*) 
$  
can be controlled by high probability bounds in the time series  setting.

Our online change detection procedure is summarized in \Cref{algorithm:main}. 
Below we briefly explain its implementation. Using dynamic programming, at the current time $j$ and for any $k\in\{1,\ldots,W\}$, with the window size $W$ such that $N_{\mathrm{train}} + W \leq j$, we maintain
\begin{align*}
    L[k]= \sum_{i=1}^{(j-W +k) -1} \widehat{\mathcal M}^{(i)}
    \quad\text{and}\quad
    R[k]= \sum_{i=(j-W +k)}^{j} \widehat{\mathcal M}^{(i)} .
\end{align*}
Hence, for a given pair $(j,k)$, the matrix $\mathcal  D\in\mathbb R^{M^p\times M^q}$ in \Cref{algorithm:main} is the CUSUM statistic
\begin{align}
    \label{eq:expression of difference matrix}
    \mathcal D &= \frac{1}{(j-W +k) -1}\sum_{i=1}^{j-W-1+k} \widehat{\mathcal M}^{(i)}\nonumber\\
    &\quad
         - \frac{1}{\,W-k+1\,}\sum_{i=(j-W +k)}^{j} \widehat{\mathcal M}^{(i)},
\end{align}
which compares the data between  the intervals $[1, (j-W +k) -1]$ and $[(j-W +k) ,  j ]$.
 For example if $j=\bb+W$ and $k=1$, from \eqref{eq:expected value of matrix at time t} we can deduce that 
\begin{align} \label{eq:D population}
    \mathcal  D &= \frac{1}{\bb}\sum_{i=1}^{\bb} \widehat{\mathcal M}^{(i)} 
      - \frac{1}{W}\sum_{i=\bb+1}^{\bb+W} \widehat{\mathcal M}^{(i)}\nonumber\\
  &\approx \mathcal M(\lambda^*)-\mathcal M(\lambda^*_a)
  = \mathcal M(\lambda^*-\lambda^*_a),
\end{align}
where {the last equality follows from the linearity of the coefficient/matricization operator $\mathcal M$ given in \eqref{eq:define M star}, that is, $\mathcal M(f-g)=\mathcal M(f)-\mathcal M(g)$ for any square-integrable functions $f$ and $g$}. To further reduce variance when estimating
$\lambda^*-\lambda_a^*$, we apply the restricted SVD procedure   to $\mathcal  D$ as described in \Cref{algorithm:restricted svd}.

\Cref{algorithm:restricted svd} has two components:  (i) zeroing out higher-order entries of $\mathcal D$ by trimming to an adaptive basis size, and (ii) applying SVD to the trimmed matrix. The trimming is adaptive to the sample size in \eqref{eq:expression of difference matrix}: as $k$ ranges from $1$ to $W$, the effective  sample size     is $W-k+1$. The necessity of trimming comes from the fact that smaller samples  only allow us to reliably estimate  a smaller number of   matrix coefficients.

We apply SVD to $\mathcal D$ because its population counterpart $\lambda^*-\lambda^*_a$ is typically approximately low rank. Since $\lambda^*-\lambda^*_a\in \lt (\Omega^{\otimes p} \times \Omega^{\otimes q})$, the functional  SVD (\Cref{theorem:svd Hilbert} in Appendix) yields
\begin{equation}
    \label{eq:SVD difference function}
    \lambda^* - \lambda^*_a = \sum_{k=1}^{\infty} \sigma_{k}(\lambda^* - \lambda^*_a) \, f_k^*(y)\, g_k^*(z),
\end{equation}
with nonincreasing singular values $\sigma_1(\lambda^* - \lambda^*_a)\ge \sigma_2 (\lambda^* - \lambda^*_a) \ge\cdots\ge 0$ such that 
$$\sum_{k=1}^\infty  \sigma_k ^2 (\lambda^* - \lambda^*_a)  =\|\lambda^*-\lambda^*_a\|_{\lt }^2<\infty,$$
and orthonormal functions
$$\{f_k^*(y) \}\subset \lt (\Omega^{\otimes p}), \quad \{g_k^*(z) \}\subset \lt (\Omega^{\otimes q}).$$   

It is a commonly used assumption   in the literature \citep[e.g.,][]{hall2006properties,raskutti2012minimax} that,   if  
$ \|\lambda^*-\lambda^*_a\|_{W^{ 2,\gamma } } <\infty$,   then  the singular values of $\lambda^*-\lambda^*_a$ decay at a 
  polynomial  or exponential  rate. 
Since $\mathcal M(\lambda^*-\lambda^*_a)$ provides an accurate matrix representation of $\lambda^*-\lambda^*_a$,
we can anticipate  that the singular values of $\mathcal M(\lambda^*-\lambda^*_a)$, and consequently the singular values  of $\mathcal D$ in \eqref{eq:D population},  decay at the same rate, see Lemma~\ref{lem:tail-projection} in Appendix for a justification.

\smallskip
\begin{remark}[Computational cost]
Due to the dynamic programming design, for  a new observation   the  computational cost of \Cref{algorithm:main}   is
$O\big(r  W^{\,1 + d/(2\gamma + p\vee q)} \big)$. 
More precisely, at time $j$, updating each matrix in the lists $L$ and $R$ costs
$O\big(W^{\,d/(2\gamma + p\vee q)}\big)$. 
Computing the rank-$r$ SVD for each difference matrix $\mathcal D$ in \Cref{algorithm:main} costs
$O\big(r W^{\,d/(2\gamma + p\vee q)}\big)$. 
Since $L$ and $R$ each contain $W$ matrices, the total   cost is
$O\big(r W^{\,1 + d/(2\gamma + p\vee q)} \big)$.
Consequently, the method is single-pass over the sequential data, and the cost per new sample   does not grow with the past time series length.
\end{remark}

\smallskip
\begin{remark}
     [PPP   change point detection in 1D]\label{remark:1d poisson}
\Cref{algorithm:main} tackles online change point detection for Poisson point process sequential data in $\mathbb X \subset \mathbb R^d$ with $d\ge 2$.  On the other hand,  in  \Cref{section:1d} of Appendix,  we present a simplified one-dimensional variant that handles PPP time series  in  $\mathbb R$ by representing the intensity as a vector rather than a matrix. The univariate setting is substantially simpler than the multivariate setting,  as  univariate nonparametric models  do  not suffer from the curse of dimensionality.
\end{remark}

\begin{algorithm}[tb]
   \caption{Online multivariate PPP change detection}
   \label{algorithm:main}
\begin{algorithmic}
   \STATE {\bfseries Input:} Smoothness parameter $\gamma>0$; dimensionality $p,q$ with $p+q=d$; rank $r$; window size $W$; threshold constant $\mathcal C_\alpha$
   
   \STATE
   \STATE {\scriptsize $\blacktriangleright$} \textbf{Initialization Stage}
   \STATE $M \leftarrow \left\lceil (W/r)^{1/(2\gamma+p\vee q)} \right\rceil$

   \FOR{$k=1$ {\bfseries to} $W$}
      \STATE $L[k] \leftarrow \sum_{i=1}^{N_{\text{train}}-W+k-1}\widehat{\mathcal M}^{(i)} {\in\mathbb R^{M^p\times M^q}}$ \quad (computed via \eqref{eq:A-i})
   \ENDFOR

   \STATE
   \STATE {\scriptsize $\blacktriangleright$} \textbf{Detection Stage}
   \STATE $\mathtt{ALARM} \leftarrow \textsc{False}$
   \FOR{$j = N_{\text{train}}+1, N_{\text{train}}+2, \ldots$}
      \FOR{$k=1$ {\bfseries to} $W-1$}
         \STATE $L[k] \leftarrow L[k+1]$
      \ENDFOR
      \STATE $L[W] \leftarrow L[W] + \widehat{\mathcal M}^{(j-1)}$

      \FOR{$k=1$ {\bfseries to} $W$}
         \STATE $R[k] \leftarrow \sum_{i=j-W+k}^{j}\widehat{\mathcal M}^{(i)}$
      \ENDFOR

      \FOR{$k=1$ {\bfseries to} $W$}
         \STATE $n_1 \leftarrow j-W-1+k$
         \STATE $n_2 \leftarrow W-k+1$
         \STATE $\mathcal D \leftarrow n_1^{-1}L[k] - n_2^{-1}R[k]$  \quad (defined in \eqref{eq:expression of difference matrix})

         \IF{$\rsvd  ( \mathcal D  ) > \mathcal C_\alpha \Big(\dfrac{r}{n_2}\Big)^{\frac{\gamma}{2\gamma+p\vee q}} \log (j)$ \quad (see \Cref{algorithm:restricted svd})
         } 
            \STATE $\mathtt{ALARM}\leftarrow\textsc{True}$
            \STATE \textbf{break}
         \ENDIF
      \ENDFOR
   \ENDFOR
\end{algorithmic}
\end{algorithm}

\begin{algorithm}[tb]
   \caption{Restricted SVD: $\rsvd(\mathcal{D})$}
   \label{algorithm:restricted svd}
\begin{algorithmic}
   \STATE {\bfseries Input:} Matrix $\mathcal D\in\mathbb R^{M^p\times M^q}$; rank $r$; sample size $n_2$; dimensions $p,q$; smoothness $\gamma>0$
   
   \STATE
   \STATE {\scriptsize $\blacktriangleright$} \textbf{Adaptive trimming}
   \STATE $m \leftarrow \left\lceil (n_2/r)^{1/(2\gamma+p \vee q)} \right\rceil$
   \STATE $\mathcal B_y \leftarrow \Big\{\phi_{i_1}(x_1)\cdots \phi_{i_p}(x_p)\Big\}_{i_1,\ldots,i_p=1}^{m}$
   \STATE $\mathcal B_z \leftarrow \Big\{\phi_{\ell_1}(x_{p+1})\cdots \phi_{\ell_q}(x_{p+q})\Big\}_{\ell_1,\ldots,\ell_q=1}^{m}$
   
   \FOR{$\mu=1$ {\bfseries to} $M^p$}
      \IF{$\Phi_\mu \notin \mathcal B_y$}
         \STATE set the $\mu$-th row $\mathcal D_{\mu,*}\leftarrow 0$
      \ENDIF
   \ENDFOR
   
   \FOR{$\eta=1$ {\bfseries to} $M^q$}
      \IF{$\Psi_\eta \notin \mathcal B_z$}
         \STATE set the $\eta$-th column $ \mathcal D_{*,\eta}\leftarrow 0$
      \ENDIF
   \ENDFOR

   \STATE
   \STATE {\scriptsize $\blacktriangleright$} \textbf{Rank-$r$ projection and score}
   \STATE $\mathcal D[r] \leftarrow \svd(\mathcal D, r)$
   \STATE {\bfseries Output:} $\|\mathcal D[r]\|_{\mathrm F}$
\end{algorithmic}
\end{algorithm}

In \Cref{theorem:main}, we provide statistical guarantees for the overall false alarm  probability      and the detection delay of \Cref{algorithm:main}.

\begin{theorem}[False-alarm control and detection delay]
\label{theorem:main}
Let the univariate basis functions $\{\phi_k\}_{k =1}^\infty $ in \eqref{eq:basis-blocks} be the Legendre polynomials.
Assume the   PPP time series $\{X^{(i)}\}_{i = 1}^\infty $ satisfies
Definition~\ref{def:model_ppp_beta_mixing} on a compact domain $\mathbb X\subset\mathbb R^d$, with $d=p+q\ge2$.
Suppose the training length $N_{\train}$ is sufficiently large.

\smallskip
\noindent\textbf{(a) No change point.}
Under scenario \textbf{(M0)} with intensity $\lambda^*$.
Choosing a sufficient large threshold constant $\mathcal C_\alpha$ in \Cref{algorithm:main},
with probability at least $1-\alpha$, \Cref{algorithm:main} never raises an alarm over the entire time horizon.

\smallskip
\noindent\textbf{(b) Single change point.}
Under scenario \textbf{(M1)} with change point $\bb\ge N_{\train}$ and intensities
$\lambda^*,\lambda_a^*$.
Let $\{\sigma_k(\lambda^*-\lambda_a^*)\}_{k\ge1}$ be the singular values of $\lambda^*-\lambda_a^*$, as specified in
\eqref{eq:SVD difference function}, and choose the rank parameter $r$ in \Cref{algorithm:main} such that
\begin{equation}\label{eq:snr for tail main}
\sqrt{\sum_{k=r+1}^\infty \sigma_k^2(\lambda^*-\lambda_a^*)}
\;\le\;
\frac{\|\lambda^*-\lambda_a^*\|_{\lt}}{5}.
\end{equation}
Let $\kappa=\|\lambda^*-\lambda_a^*\|_{\lt}$ and define
\begin{equation}\label{eq:detection delay bound main}
\Delta \;=\; \Big\lceil C_{\lag}\, r \,\big(\log(\bb)/\kappa\big)^{\,2+(p\vee q)/\gamma}\Big\rceil,
\end{equation}
where $C_{\lag}$ is a sufficiently large constant depending only on $\mathcal C_\alpha$.
If the window size satisfies $W\ge \Delta$, then with probability at least $1-\bb^{-3}$,
\Cref{algorithm:main} raises an alarm within the time interval $(\bb,\bb+\Delta]$.
\end{theorem}

In \Cref{theorem:main} \textbf{(a)}, $\mathcal{C}_\alpha$ is the threshold parameter that controls the overall false-alarm probability.
We discuss a data-driven strategy for calibrating $\mathcal{C}_\alpha$ in \Cref{remark:tuning parameters}.

\begin{remark}[Generality of \Cref{theorem:main} across coordinate partitions]
\label{remark:partition scope}
The condition~\eqref{eq:snr for tail main} in \Cref{theorem:main} follows from a commonly used functional regularity condition in the literature. In particular, we show in \Cref{lemma:eigen decay} in the Appendix that, if the singular values $\{\sigma_k(\lambda^*-\lambda^*_a)\}_{k=1}^\infty$ decay at a polynomial (with degree $>1/2$) or exponential rate, then~\eqref{eq:snr for tail main} holds with a constant rank $r$. Such functional singular-value decay assumptions are standard in the literature \citep[e.g.,][]{hall2006properties,raskutti2012minimax}, and they accommodate separable, additive, and finite-rank intensity classes as well as more general smooth interactions. The factor $5$ in the denominator of~\eqref{eq:snr for tail main} is chosen for convenience and we do not attempt to optimize it.

Importantly, \Cref{theorem:main} establishes the false-alarm and detection-delay guarantees for \emph{any} coordinate partition $[d]=\mathcal I_1\cup\mathcal I_2$ for which the matricized intensity satisfies~\eqref{eq:snr for tail main}; the proof does not rely on any specific property of the partition-selection rule. The empirical cross-group-correlation criterion in \Cref{remark:tuning parameters} is therefore a practical and numerically stable heuristic for locating such a partition rather than a structural requirement of the theorem; robustness to alternative rank choices and to random coordinate partitions is verified empirically in our numerical studies.
\end{remark}

\Cref{theorem:main} also implies that, when a change point is present, the detection delay is at most   $ O(\kappa^{-2-(p\vee q)/\gamma})$. By comparison, \citet{padilla2021optimal} and \citet{madrid2023change}   show that, for  nonparametric density change point detection in the offline setting, the error scales as $O\big(\kappa^{-2-d/\gamma}\big)$, where $\kappa$ is the size of the change in  $\lt $-norm, $\gamma$ is the degree of smoothness, and $d$ is the ambient dimension of the density function. Although we study Poisson intensity changes for time series data, our detection delay bound is strictly better in order. This is because   $p+q=d$, and thus   $p\vee q<d$.   In the difficult regime where the change size is small, $\kappa\to  0$, we have
$$  \kappa^{-2-(p\vee q)/\gamma} \ll \kappa^{-2-d/\gamma}. $$

\section{Numerical  Studies}
\label{section:simulation}
We designed simulated experiments to evaluate the performance of \Cref{algorithm:main}, which we refer to as the \textbf{Matrix detector}, against two benchmarks. First, the \textbf{MMD detector} follows \citet{li2015m}: it represents the data in each window via its empirical kernel mean embedding and computes a blockwise maximum mean discrepancy (MMD) statistic between the point samples from the pre- and post-change segments. Second, the \textbf{KIE detector} is adapted from the density change point detector of \citet{madrid2023change}: it estimates the pre- and post-change intensities using a kernel intensity estimator (KIE) and then forms a CUSUM statistic based on the discrepancy between the two estimated intensities, measured in the $\lt$-norm. A Python implementation of our method is available at \href{https://github.com/HaotianXu/Online-Change-Point-Detection-for-Multivariate-PPP}{the GitHub repository}.
\subsection{Selection of Tuning Parameters}
 \label{remark:tuning parameters}
\noindent  \textbf{MMD detector.} For the MMD detector, we adopt the default tuning-parameter selection in \citet{li2015m} and use a Gaussian kernel throughout the numerical experiments.

\noindent  \textbf{KIE detector.} We follow the default choice in \citet{madrid2023change}, and select the kernel bandwidth as well as the detection threshold using cross-validation on the training data.

\noindent  \textbf{Matrix detector.}
We set the smoothness parameter $\gamma=2$, meaning the intensity functions are at least twice differentiable, which is standard in the nonparametric literature \citep[e.g.,][]{wasserman2006all}.
We partition the $d$ coordinates into two groups using a correlation-based split criterion.
Let $\corr(i,j)$ denote the empirical  correlation between the $i$-th and $j$-th coordinate
of     all points   in the training data.
For any partition $(A,B)$ of $\{1,\ldots,d\}$, define
\[
\Delta(A,B) \;=\; \frac{1}{|A||B|}\sum_{i\in A}\sum_{j\in B} \big | \corr(i,j) \big| .
\]
We select the coordinate partition used in \Cref{algorithm:main} by searching over all nontrivial partitions $(A,B)$ and choosing
\[
(A^*,B^*)\in \argmin_{(A,B)} \; \Delta(A,B),
\]
and then set $p=|A^*|$ and $q=|B^*|$ such that  $p+q=d$.

Given the coordinate partition  and hence $p$ and $q$, we next select the rank parameter $r$ using sample splitting and a goodness-of-fit  criterion.
Without loss of generality, assume $N_{\train}$ is even.
For any candidate rank $r$, define
\begin{align*}
    \widehat{\mathcal W}_{1}
&= 
\frac{2}{N_{\train}}\sum_{i=1}^{N_{\train}/2}\widehat{\mathcal M}^{(i)}
 \text{ and} 
 \\
\widehat{\mathcal W}_{2}
 &= 
\frac{2}{N_{\train}}\sum_{i=N_{\train}/2+1}^{N_{\train}}\widehat{\mathcal M}^{(i)},
\end{align*}

where $\widehat{\mathcal M}^{(i)}$ is defined in \eqref{eq:A-i}.
Let $\svd(\mathcal W,r)$ denote the rank-$r$ truncated SVD approximation of $\mathcal W$.
We choose $r$ by
\[
\widehat r \in \argmin_{ r\in \{1,\ldots,\min(M^p,M^q)\}}\;
\Big\|\svd(\widehat{\mathcal W}_{1},r)-\svd(\widehat{\mathcal W}_{2},r)\Big\|_{\mathrm F}.
\]

 The last parameter to select is the threshold $\mathcal C_\alpha$.
 {We calibrate $\mathcal C_\alpha$ using a block-permutation procedure on the training data. Fix a block length $L_{\mathrm{block}}$ and let $K_{\mathrm{block}} = \lfloor N_{\train}/L_{\mathrm{block}} \rfloor$. We partition the first $K_{\mathrm{block}} L_{\mathrm{block}}$ training windows into consecutive blocks
\[
\mathcal B_\ell
=
\{(\ell-1)L_{\mathrm{block}}+1,\ldots,\ell L_{\mathrm{block}}\},
\quad \ell=1,\ldots,K_{\mathrm{block}}.
\]
For each repetition $b=1,\ldots,500$, we randomly permute the $K_{\mathrm{block}}$ blocks and split the permuted blocks into two groups of sizes $\lfloor K_{\mathrm{block}}/2\rfloor$ and $K_{\mathrm{block}}-\lfloor K_{\mathrm{block}}/2\rfloor$ (equal when $K_{\mathrm{block}}$ is even, and differing by one block otherwise).
Let $\mathcal J_{1,b}$ and $\mathcal J_{2,b}$ denote the indices of the first and second block groups, and define the corresponding time-index sets
\[
\mathcal I_{1,b}=\bigcup_{\ell\in \mathcal J_{1,b}}\mathcal B_\ell,
\qquad
\mathcal I_{2,b}=\bigcup_{\ell\in \mathcal J_{2,b}}\mathcal B_\ell.
\]
We then compute the corresponding two-sample CUSUM matrix
\[
\widehat{\mathcal D}_b
=
\frac{1}{|\mathcal I_{1,b}|}\sum_{i\in \mathcal I_{1,b}}\widehat{\mathcal M}^{(i)}
-
\frac{1}{|\mathcal I_{2,b}|}\sum_{i\in \mathcal I_{2,b}}\widehat{\mathcal M}^{(i)},
\]
and record the test statistic $\|\widehat{\mathcal D}_b\|_{\mathrm F}$.
Finally, we set $\mathcal C_\alpha$ to be the $(1-\alpha)$-quantile of
\[
\left\{
\frac{\|\widehat{\mathcal D}_b\|_{\mathrm F}}
{\left(\widehat r / (|\mathcal I_{1,b}|\wedge |\mathcal I_{2,b}|)\right)^{\gamma/(2\gamma+p\vee q)}\log(N_{\train})}
\right\}_{b=1}^{500},
\]
and use it as the input threshold in \Cref{algorithm:main}.}

\subsection{3D Intensity with Temporally Dependent Latent Variables}
We generate Poisson point process    time series data in $d=3$ dimensions with a change point at time $\bb=1200$.
We sample each process window using the thinning algorithm \citep{lewis1979simulation}. Specifically, for $t\le \bb$,
the intensity function is
\[
\lambda_t^*(x)=   z^+ _{t,1}  \prod_{j=1}^3\{\sin(x_j)+1\}
+ z^+_{t,2}\prod_{j=1}^3\{\cos(x_j)+1\},
\]
while for $t>\bb$ the intensity function is
\[
\lambda_t^*(x)= z^+_{t,1}\prod_{j=1}^3 \exp(-x_j^2)
+ z^+_{t,2}\prod_{j=1}^3 x_j .
\]
Here $x\in[0,1]^3$, $z_t^+ =\max\{z_t,0 \}$, and $z_t=(z_{t,1},z_{t,2})\in\mathbb R^2$ follows an autoregressive model
\[
z_{t+1}=
\begin{bmatrix}0.5&0.1\\0.1&0.5\end{bmatrix}z_t+\epsilon_t,
\qquad
\epsilon_t\overset{i.i.d.}{\sim}N([3,1],I_2),
\]
with $I_2$ the two-dimensional identity matrix.  
Each run consists of $N_{\train}=1000$ pre-change samples and $N_{\text{total}}=1500$ samples in total.
All approaches are initialized on the same training data and evaluated over 100 Monte Carlo replications.
In \Cref{tab:sim_3d_results}, we summarize the performance of the three methods using the default tuning-parameter
choices described in \Cref{remark:tuning parameters}.
 
 \begin{table}[tb]
\centering
\caption{Simulation results in the 3D setting over 100 replications. False Alarm denotes an alarm raised at or before $\bb$, Correct Detection denotes an alarm raised in $(\bb, N_{\text{total}}]$, and No Alarm denotes no alarm by $N_{\text{total}}$. The average detection delay (ADD) and the corresponding standard deviation (SD) are reported conditional on correct detection.}
\label{tab:sim_3d_results}
\begin{tabular}{lccc}
\toprule
 & Matrix & MMD & KIE \\
\midrule
False Alarm       & 6\%  & 5\%  & 6\% \\
Correct Detection & 94\% & 93\% & 89\% \\
No Alarm          & 0\%  & 2\%  & 5\% \\
ADD (SD)          & 9.19 (3.80) & 40.39 (29.79) & 43.98 (34.63) \\
\bottomrule
\end{tabular}
\end{table}


To further understand the performance of our proposed method, we evaluate  $15$ threshold values per method, ranging from  low detection sensitivity  to high. Performance was summarized using two metrics: (i) false alarm probability
(FAP) and (ii) average detection delay (ADD). If no alarm occurs by $N_{\text{total}}$, we set the delay to
$N_{\text{total}}-\bb$.

\begin{figure}[ht]
  \centering
  \includegraphics[scale=0.37]{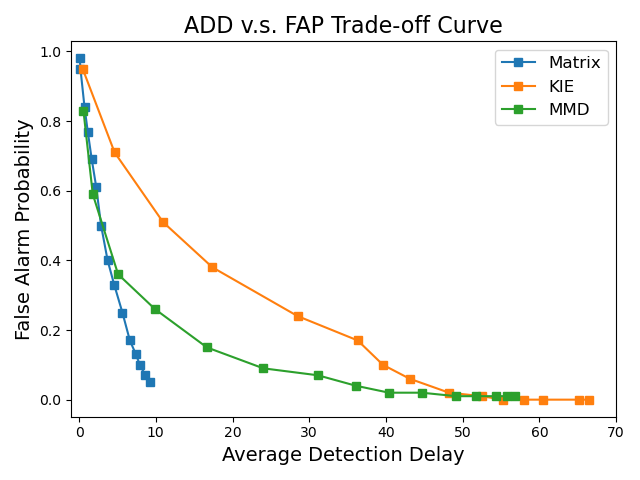}
  \caption{FAP vs.\ ADD comparison among three detectors under the 3D simulation setting.}
  \label{fig:simu_tradeoff}
\end{figure}

Figure~\ref{fig:simu_tradeoff} presents the trade-off between empirical FAP and ADD. Our method (\textbf{Matrix}) generally achieves a smaller average detection delay than the other two methods at comparable false alarm probabilities.

  \subsection{4D Intensity with Autoregressive Scale}
\label{section:4d}
We generate Poisson point process time series data in $d=4$ dimensions with a change point at time $\bb=1200$.
The intensity function is
\[
\lambda_t^*(x)= y_t^+\bigg\{\prod_{j=1}^4 2x_j^3 + \prod_{j=1}^4 2\exp(-x_j)\bigg\},
\qquad x\in[0,1]^4,
\]
where the scale sequence $\{y_t\}_{t= 1}^\infty $ follow the autregressive model
\[
y_{t+1}= 
\begin{cases}
    0.1\cdot |X^{(t)}|+8+\epsilon_t  &\text{for } t\le \bb\\
    0.1\cdot |X^{(t)}|+4+\epsilon_t \ &\text{for } t>\bb
\end{cases},
\]
with $\epsilon_t \overset{i.i.d.}{\sim} N(0,1)$, and $y^+_t =\max\{ y_t, 0\} $. Here $X^{(t)}$ denotes the Poisson point process observed at time $t$
with intensity $\lambda_t^*$, and $|X^{(t)}|$ is its cardinality.
Each run consists of $N_{\train}=1000$ pre-change samples and $N_{\text{total}}=1500$ samples in total.
All approaches were initialized on the same training data and evaluated over 100 Monte Carlo replications.
In \Cref{tab:sim_4d_results}, we summarize the numerical performance of the three methods based on the default
tuning-parameter choices described in \Cref{remark:tuning parameters}.
 \begin{table}[ht]
\centering
\caption{Simulation results in the 4D setting over 100 Monte Carlo replications. The average detection delay (ADD) and the corresponding standard deviation (SD) are reported conditional on correct detection.}
\label{tab:sim_4d_results}
\begin{tabular}{lccc}
\toprule
 & Matrix & MMD & KIE \\
\midrule
False Alarm       & 3\%  & 4\%  & 3\% \\
Correct Detection & 97\% & 94\% & 92\% \\
No Alarm          & 0\%  & 2\%  & 5\% \\
ADD (SD)          & 13.44 (5.21)  & 20.83 (10.90)   & 31.50 (21.88)  
\\
\bottomrule
\end{tabular}
\end{table}

\begin{figure}[ht]
  \centering
  \includegraphics[scale=0.4]{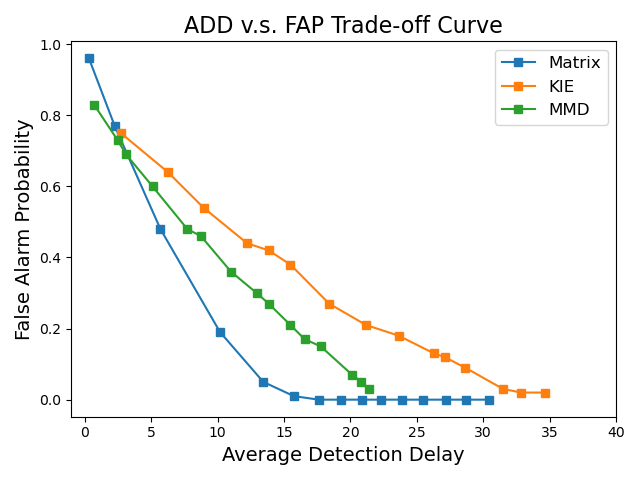}
  \caption{FAP vs.\ ADD comparison among three detectors under the 4
  D simulation setting.}
  \label{fig:simu_tradeoff_4d}
\end{figure}

Figure~\ref{fig:simu_tradeoff_4d} presents the trade-off between empirical false alarm probability (FAP) and
average detection delay (ADD). Our method achieves substantially lower detection delays than
the other two methods at comparable false alarm probabilities.

\subsection{A Real Data Example}
Poisson point processes are widely used to model earthquake activity \cite{anagnos1988review,cornell1968engineering,wang2006understanding}.
We test our online change detection method on an earthquake dataset from the state of Oklahoma covering January 2000 to December 2018,
obtained from the \cite{usgs_comcat_2025}.
The dataset records each earthquake's longitude, latitude, depth, and magnitude, and we aggregate events by week so that each month forms one
Poisson point process observation.
We use data from January 2000 to December 2007 for training.

\begin{figure*}[t]
  \centering
  \includegraphics[width=\textwidth,trim=0 10 0 10,clip]{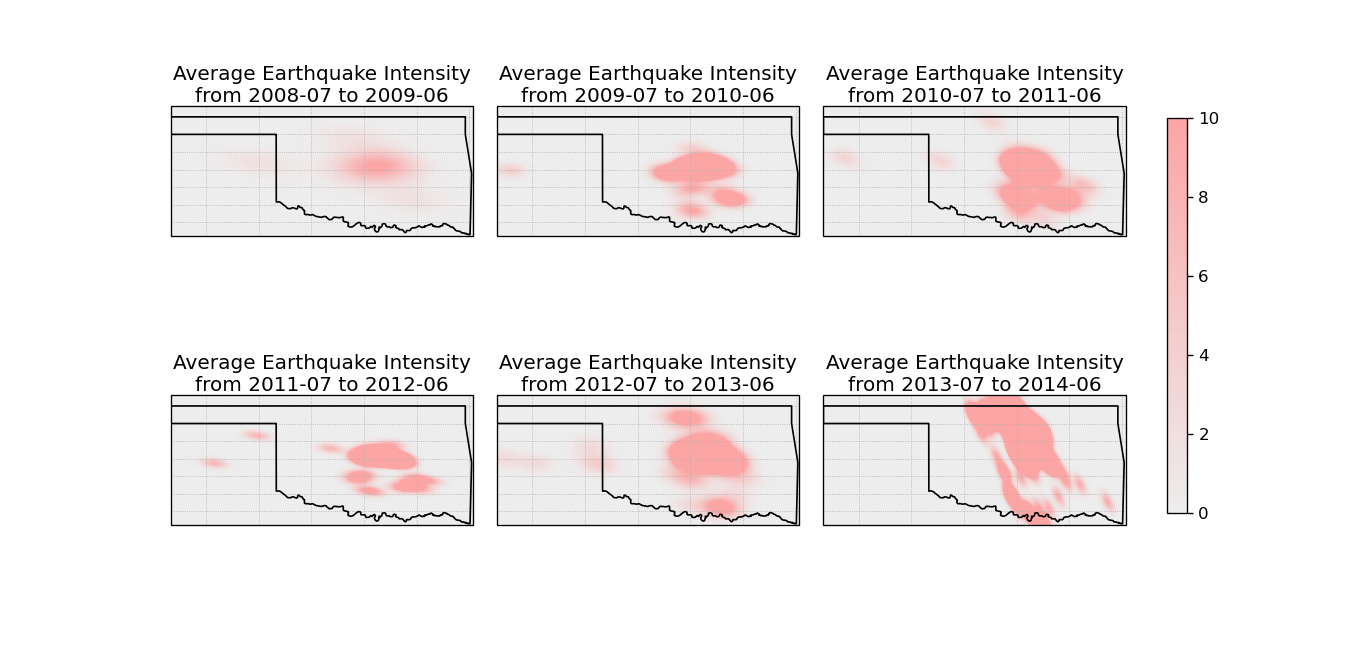}
  \caption{{Estimated yearly average earthquake intensity in Oklahoma over six consecutive 12-month windows from July 2008 through June 2014. The change point detected by \Cref{algorithm:main} is June 2009, which lies between the first two panels. The panels reveal a clear and progressive increase in seismic activity in the years following the detected change.}}
  \label{fig:ek}
\end{figure*}
 In this real-data experiment, we select all tuning parameters according to the procedure described in
\Cref{remark:tuning parameters}. Our online procedure raises an alarm in   June 2009.
This timing is consistent with reports by U.S. national agencies indicating that earthquake activity in Oklahoma,
particularly for events with magnitude below 3, increased rapidly beginning around 2009 and has been linked to
wastewater disposal and injection practices; see \cite{usgs_oklahoma_nodate}.
{Figure~\ref{fig:ek} visualizes this trend through a time series of estimated yearly average intensity surfaces over six consecutive 12-month windows; the detected change point lies between the first two windows, and the seismic-activity hotspots expand markedly in the windows that follow.} In contrast, the \textbf{KIE} detector raises an alarm in July 2007,
while the \textbf{MMD} detector  does not raise an alarm by December 2018.

\section{Conclusion}

We studied online change point detection for multivariate inhomogeneous Poisson point processes (PPP) time series under the $\beta$-mixing  temporal dependence assumption. Our approach maps each PPP realization to a finite-dimensional intensity matrix via an orthonormal basis expansion, and then leverages approximate low-rank structure through a restricted SVD. This representation yields a single-pass detection algorithm with constant per-observation cost and total runtime linear in the number of windows.  Theoretically, we develop a framework  jointly capture the approximation bias and stochastic variance and provides explicit guarantees for both false-alarm control and detection delay. Empirically, our method demonstrates strong detection performance and robustness across a range of regimes, while maintaining computational efficiency that scales linearly in the time series length.

\newpage

\bibliographystyle{plainnat}
\bibliography{icml_paper}

\newpage
\appendix
\section{Proof of \Cref{theorem:main}}

\noindent\textbf{Setup and notation.}
Let $\{X^{(i)}\}_{i\ge1}$ be a sequence of PPPs satisfying
\Cref{def:model_ppp_beta_mixing} on a compact domain $\mathbb X\subset\mathbb R^d$.
Let $\{\phi_k\}_{k\ge1}$ be an orthonormal basis of $\lt(\Omega)$ such that
$\|\phi_k\|_\infty \le C_\phi$ for all $k$, where $0<C_\phi<\infty$ is an absolute constant.
Fix integers $p,q\ge1$ with
\[
p+q=d
\qquad\text{and}\qquad
p\ge q.
\]
For $x=(x_1,\ldots,x_d)\in \mathbb X = \Omega^{\otimes d}$, write the coordinate split
\[
x=(y,z),
\qquad
y=(x_1,\ldots,x_p)\in\Omega^{\otimes p},
\qquad
z=(x_{p+1},\ldots,x_{p+q})\in\Omega^{\otimes q}.
\]
For an integer $m\ge1$, define the tensor-product bases
\begin{align*}
\{\Phi_\mu(y)\}_{\mu=1}^{m^p}
&= \Big\{\phi_{\mu_1}(x_1)\cdots\phi_{\mu_p}(x_p)\Big\}_{\mu_1,\ldots,\mu_p=1}^m,\\
\{\Psi_\eta(z)\}_{\eta=1}^{m^q}
&= \Big\{\phi_{\ell_1}(x_{p+1})\cdots\phi_{\ell_q}(x_{p+q})\Big\}_{\ell_1,\ldots,\ell_q=1}^m.
\end{align*}

For each window $i$ and basis size $m$, define the empirical intensity matrix
$\widehat A^{(i),m}\in\mathbb R^{m^p\times m^q}$ entrywise by
\begin{equation*}
\widehat A^{(i),m}_{(\mu,\eta)}
=
\sum_{x=(y,z)\in X^{(i)}} \Phi_\mu(y)\,\Psi_\eta(z).
\end{equation*}
Given integers $1\le t < n$, define the (two-sample) CUSUM matrix
$\widehat A^{m}_{n,t}\in\mathbb R^{m^p\times m^q}$ by
\begin{equation*}
\widehat A^{m}_{n,t}
=
\frac{1}{t}\sum_{i=1}^{t}\widehat A^{(i),m}
\;-\;
\frac{1}{n-t}\sum_{i=t+1}^{n}\widehat A^{(i),m},
\qquad
\widehat A^{m}_{n,t}[r] = \svd(\widehat A^{m}_{n,t},r).
\end{equation*}

In our procedure, the basis size is chosen adaptively as a function of $(n,t)$:
\begin{equation}
\label{eq:size parameter}
m_{n,t}
=
\Big\lceil \Big(\frac{n-t}{r}\Big)^{1/(2\gamma+p)} \Big\rceil.
\end{equation}
For brevity, we write
\begin{equation}
\label{eq:cusum poisson general}
\widehat A_{n,t} = \widehat A^{m_{n,t}}_{n,t},
\qquad
\widehat A_{n,t}[r] = \svd(\widehat A_{n,t},r).
\end{equation}
Finally, define the threshold
\begin{equation}
\label{eq:threshold parameter}
\tau_{n,t}
=
C_\alpha
\Big(\frac{r}{n-t}\Big)^{\gamma/(2\gamma+p)}\log n,
\end{equation}
where $C_\alpha>0$ is a sufficiently large constant depending only on $\alpha$.

\begin{proof}[Proof of \Cref{theorem:main}]
      \Cref{theorem:main} directly follows  \Cref{lemma:type 1} and \Cref{lemma:type 2}.
\end{proof}

\begin{lemma} \label{lemma:type 1}
    Let $m_{n,t}$, $\widehat A_{ n,t }[r]$ and $\tau_{n,t }$ be defined in \eqref{eq:size parameter}, \eqref{eq:cusum poisson general} and \eqref{eq:threshold parameter}, respectively. Suppose $\{ X^{(i)} \}_{i=1}^{\infty }$ are sampled from \textbf{(M0)} with the common intensity function $\lambda^*$. Then 
    $$ \p \big( \|\widehat A_{ n,t }  [r] \|_\F \ge \tau_{n,t} \text{ for all } 1\le t < n <\infty\big) \le 1-\alpha  .$$
\end{lemma}
\begin{proof}
Since there is no change points, we have $\delta^* = \lambda^* - \lambda_a^* = 0$. 
We apply  \Cref{lemma:main deviation} to any give $t, n$ and have that with probability at most $2\epsilon n^{-3}$,
$$ \| \widehat A _{n,t}[r] \|_\F \ge C_2  m_{n,t}^{-\gamma  } + C_\epsilon \sqrt r\{ \log^2 (m_{n,t})+\log^2(n)  \} \sqrt {\frac{m_{n,t}^{p}}{n-t} }  , $$
    where we have used  $\sigma_r(\delta^*)=0$ when $\delta^*=0$. Since 
    $m_{n,t}=  \lceil \big( \frac{n-t }{r}  )^{1/(2\gamma +p)} \rceil $, it follows that
    $$ \p\bigg( \| \widehat A_{n,t}  [r] \|_\F \ge C_\epsilon^{\prime}  \big( \frac{r}{n-t}   \big )^{ \gamma /(2\gamma +p)} \log^2(n)   \bigg)  \le 2\epsilon n^{-3}.$$
    By a  union bound,   for any  $  n\in \mathbb Z^+$ it holds that 
     $$ \p\bigg( \| \widehat A_{n,t} [r]\|_\F \ge C_\epsilon'  \big( \frac{r}{n-t}\big )^{ \gamma /(2\gamma +p)}  \log(n)   \text{ for all } 1\le t < n \bigg)  \le 2\epsilon n^{-2} .$$
     Since 
     $\sum_{n=1}^\infty n^{-2} = \pi^2/6,$ if follows that if 
     $\epsilon= 3\alpha  /\pi^2$,  then by another union bound, 
     $$ \p\bigg( \| \widehat A _{n,t}[r]\|_\F \ge C_\alpha    \big( \frac{r}{n-t}\big )^{ \gamma /(2\gamma +p)}   \log(n)   \text{ for all } 1\le t < n <\infty  \bigg)  \le \alpha .$$
     The desired result follows from the choice of $\tau_{n,t}$ in \eqref{eq:threshold parameter}. 
\end{proof}
\begin{lemma}\label{lemma:type 2}
    Let  $\tau_{n,t }$ be defined in \eqref{eq:threshold parameter}. Suppose $\{ X^{(i)} \}_{i=1}^{\bb  } \cup \{ X^{(i)} \}_{i=\bb+1}^{\infty }$ are  sampled from \textbf{(M1)} with true change point $\bb$, pre-change intensity function $\lambda^*$ and post-change intensity $\lambda^*_a$. 
    Suppose   that $p\ge q$ and that 
    \begin{align}
        \label{eq:snr for tail}\sqrt{ \sum_{ k=r+1}^\infty \sigma_k^2 (\lambda^* -\lambda^*_a)}  \le 
    \frac{\| \lambda^* -\lambda^*_a \|_\lt}{5}. 
    \end{align}   
 Let $ \kappa = \|\lambda^* -\lambda^*_a \|_\lt $.    Suppose  in addition that 
 \begin{align}
     \label{eq:detection delay bound}\Delta \ge  C_{ \lag} r ( \log(\bb)/\kappa)^ {2+p/\gamma } 
 \end{align}     
    where $C$ is a sufficiently large constant only depending on $C_\alpha $ in \eqref{eq:threshold parameter}. 
    Then  
    $$ \p \big( \|\widehat A_{ \Delta+\bb,\bb  }[r] \|_\F \ge \tau_{\Delta+\bb,\bb}  \big) \ge 1-\bb^{-3},$$
    where $\widehat A_{\Delta+\bb ,\bb}[r]$ is defined according  to \eqref{eq:cusum poisson general}.
\end{lemma}
\begin{proof}
    We apply  \Cref{lemma:main deviation} to  
    $  \widehat A_{ \Delta+\bb,\bb  }[r]
    $  with the difference of ground truth intensity functions being $\delta^*= \lambda^* -\lambda^*_a$ to deduce  that,  with probability at least  $1- 2\epsilon \bb^{-3}$,
 \begin{align} \nonumber
     &\bigg | \|\widehat A_{ \Delta+\bb,\bb  }[r] \|_\F    -  \| \delta^*\|_\lt \bigg| 
     \\ \nonumber
     \le &   C_2 m _{\Delta+\bb,\bb }^{-\gamma  } + C_\epsilon  \sqrt r\big\{\log(\bb) +\log( m_{\Delta+\bb,\bb }) \big\}     \sqrt { \frac{m^p  _{\Delta+\bb,\bb } }{\Delta }}    +  4\sqrt { \sum_{k=r+1}^\infty \sigma^2_k( \delta^*)}  ,
     \\ \label{eq:alternative term 0}
     \le & C_2 m _{\Delta+\bb,\bb }^{-\gamma  } + 2 C_\epsilon   \sqrt r \log(\bb)       \sqrt { \frac{m^p  _{\Delta+\bb,\bb }  }{\Delta}}    +  4\sqrt { \sum_{k=r+1}^\infty \sigma^2_k( \delta^*)}   
 \end{align}  
    Here  the second inequality follows from assumption that 
    $N_{\train}$ is sufficiently large, $\bb \ge N_{\train}$, and the observation that 
    $$  m_{\Delta+\bb,\bb } =\lceil  \big(  \Delta  /r  )^{1/(2\gamma +p)}\rceil \le \bb . $$
    It follows that with probability at least  $1- 2\epsilon \bb^{-3}$, 
    \begin{align}\nonumber
        \|\widehat A_{ \Delta+\bb,\bb  }[r] \|_\F   \ge   &\frac{\| \delta^*\|_\lt}{5}   - \bigg( 
      C_2 m^{-\gamma  } _{\Delta+\bb,\bb }  + 2 C_\epsilon   \sqrt r \log(\bb)       \sqrt { \frac{m^p  _{\Delta+\bb,\bb } }{\Delta}}      \bigg) 
      \\\nonumber
      = & \frac{\kappa}{5}   - \bigg( 
      C_2 m^{-\gamma  }  _{\Delta+\bb,\bb }+ 2 C_\epsilon   \sqrt r \log(\bb)       \sqrt { \frac{m^p  _{\Delta+\bb,\bb } }{\Delta}}      \bigg)  
      \\ \label{eq:alternative term 1}
      \ge & \frac{\kappa}{5} -C_3   \log(\bb) (r /\Delta) ^{\gamma /(2\gamma+p )} 
      > \frac{\kappa}{6},
 \end{align}  
 where the first inequality follows from  \eqref{eq:alternative term 0} and \eqref{eq:snr for tail}, the second inequality follows from $m_{\Delta+\bb,\bb } =\lceil  \big(  \Delta  /r  )^{1/(2\gamma +p)}\rceil$, and the last inequality follows from
 \eqref{eq:detection delay bound} with sufficiently large constant $C_{\lag}$. Note that 
 \begin{align}
     \label{eq:alternative term 2}\tau_{\Delta+\bb,\bb} = C_\alpha \bigg( \frac{r}{\Delta}\bigg)^{\gamma/(2\gamma+p)} \log(\Delta+\bb) \le  2C_\alpha  \bigg( \frac{r}{\Delta}\bigg)^{\gamma/(2\gamma+p)} \log( \bb) \le \kappa/6.
 \end{align} 
Here the equality follows from  \eqref{eq:threshold parameter}, the first  inequality follows from the fact that $\Delta \le \bb $, and the second inequality follows from \eqref{eq:detection delay bound} with sufficiently large $C_{\lag}$. The desired result follows from \eqref{eq:alternative term 1}
and \eqref{eq:alternative term 2}. \end{proof}

\section{Deviation Bounds}
\begin{theorem}[Singular value decomposition in function space]
\label{theorem:svd Hilbert}
Let $F   (y, z) : \mathbb R^{ p } \times \mathbb R^{ q } \to \mathbb R $ be any function such that $\|F \|_{\lt( \mathbb R^{ p + q })}   <\infty $.   There exists   a collection of     singular values   $  \sigma_1  (F )  \ge \sigma_2(F) \ge \cdots \ge 0 $, and two collections of orthonormal basis functions $\{ f_\ri(y)  \}_{\ri=1}^{\infty }  \subset \lt(\mathbb R^p)  $ and $\{ g_\ri   (z)  \}_{\ri=1}^{\infty } \subset \lt(\mathbb R^q) $  such that 
\begin{align}\label{eq:svd Hilbert}F  (y, z) = \sum_{\ri=1}^{\infty} \sigma_{\ri} (F ) f _\ri  (y) g_\ri(z).  
\end{align} 
\end{theorem} 
\begin{proof} See  Section 6 of \cite{brezis2011functional}.
\end{proof}

     Let $\delta^* (x) = \lambda^*(x) - \lambda^*_a(x)  : \mathbb R^{d } \to \mathbb R $.  Suppose the SVD of $ \delta^* $ satisfies 
     $$\delta^*(x) = \delta^* (y,z) = \sum_{\rho=1}^\infty \sigma_\rho( \delta^* ) f^*_\rho (y) g^*_\rho(z). $$
     For a positive integer $r$, let
     $$  \delta^* [r] (y,z) = \sum_{\rho=1}^r \sigma_\rho( \delta^* ) f^*_\rho (y)  g^*_\rho(z).$$
     Therefore $ \delta^* [r]   $ is the best rank-$r$ estimate of $ \delta^* $.
We can represent  $ \delta^* $ as a matrix in the following way.  
Let $A^*\in \mathbb R^{m^p\times m^q}$ be the coefficient matrix whose $(\mu, \eta)$ entry is
\begin{align}
    \label{eq:define A star}A^* _{(\mu,\eta)} = \iint_{\mathbb R^{p+q}} \delta^*(y,z) \Phi_\mu(y) \Psi_\eta(z) dydz. 
\end{align}
We approximate $\delta^*(y,z)$ by
\begin{align} \label{eq:density finite matrix}
    \delta_m^*(y,z) = \sum_{\mu=1}^{m^p} \sum_{\eta=1}^{m^q} A_{(\mu,\eta)}^* \Phi_{\mu} (y)\Psi_{\eta}(z).
\end{align} 
It was shown in Appendix G1 of \cite{Peng2024VRS-article} that if $\{ \phi\}_{k=1}^\infty$ are chosen to be the univariate Legendre polynomial bases,   then 
\begin{align}
    \label{eq:approximation errors}\| \delta^*-\delta^*_m \|_{\lt} \le C \| \delta^*\|_{W^{2,\gamma }} \cdot m^{-\gamma } 
\end{align}  
where $\| \delta^*\|_{W^{2,\gamma }}$ is the Sobolev norm of $\delta^*$.
\\

\begin{definition} \label{eq:definition appendix model}
    Let   $N_1, N_2 \in \mathbb Z_+$ be such that $N_1+ N_2 \le  N$.  Let $\{ X^{(i)} \}_{i=1}^{N_1+ N_2}$ be a collection of  PPPs satisfying \Cref{def:model_ppp_beta_mixing}. Suppose that  the  intensity function of $\{ X^{(i)} \}_{i=1}^{N_1 } $ is 
$$\lambda^*(x) : \mathbb R^d \to \mathbb R^+,$$ 
and the  intensity function of $\{ X^{(i)} \}_{i=N_1+1}^{N_1+N_2 } $ is $$\lambda^*_a(x) : \mathbb R^d \to \mathbb R^+.$$  
\end{definition}
Let $\widehat A^\ii  \in \mathbb  R^{m^p\times m^q}$ be such that 
 $$\widehat A^\ii_{(\mu,\eta)}    = \sum_{x=(y,z)\in X^\ii }\Phi_\mu (y)     \Psi_\eta (z) .$$  
Define 
\begin{align}
    \label{eq:cusum poisson}\widehat A= \frac{1}{N_1}\sum_{i=1}^{N_1} \widehat A^\ii -  \frac{1}{N_2}\sum_{i=N_1+1}^{N_1+N_2} \widehat A^\ii  .
\end{align}  
With  $x=(x_1,\ldots, x_d)$,  $ y=(x_1, \ldots,x_{ p }) $  and $z = (x_{ p +1},\ldots, x_{ p + q })$, we can write 
\begin{align}
    \label{eq:define hat p}{\widehat \delta} (y, z)   = \sum_{\mu=1}^{m^p} \sum_{\eta=1}^{m^q} \widehat A_{ (\mu, \eta) } \Phi_\mu(y) \Phi_{\eta}(z).
\end{align}

\begin{lemma}
    Let ${\widehat \delta} $ be defined in \eqref{eq:define hat p}. Then $$ \E \widehat \delta = \delta^*_m. $$
\end{lemma}
 \begin{proof}
 Note that by \Cref{lemma:campbell}, when $i\le N_1$,
 $$ \E (\widehat A^\ii_{( \mu, \eta ) })  = \iint_{\mathbb R^{p+q}} \Phi_\mu (y)\Psi_\eta (z) \lambda^* (y,z) dydz.  $$
 Therefore 
 \begin{align}
     \label{eq:A hat A star} \E (\widehat A_{ (\mu,\eta) }) = \iint_{\mathbb R^{p+q}} \Phi_\mu (y)\Psi_\eta (z) (\lambda^* (y,z) -\lambda^*_a(y,z) )  dydz =A^*_{\mu, \eta},
 \end{align}
    where the last equality follows from  \eqref{eq:define A star}. Therefore 
    $$ \E({\widehat \delta} (y,z) ) =\sum_{\mu=1}^{m^p} \sum_{\eta=1}^{m^q} \E( \widehat A_{ (\mu, \eta) } ) \Phi_\mu(y) \Phi_{\eta}(z)   =   \sum_{\mu=1}^{m^p} \sum_{\eta=1}^{m^q}  A^* _{(\mu, \eta) }   \Phi_\mu(y) \Phi_{\eta}(z)  = \delta^*_m(y, z). $$
 \end{proof}

\begin{lemma} \label{lemma:main deviation}
    Let  $\widehat A \in \mathbb  R^{m^p\times m^q} $  be defined in \eqref{eq:cusum poisson} and   $$ \widehat A[r] = \svd(\widehat A, r).$$ Suppose in   that there exists an absolute constant $C_1$ such that 
    $$ \max\{ \| \lambda^*\|_\infty, \| \lambda^*_a\|_\infty\}  < C_1 \quad  \max\{ \| \lambda^*\|_ {W^{2,\gamma  }}, \| \lambda^*_a\|_ {W^{2,\gamma  }}\}<C_1,$$
    and that 
      \begin{align}
          \label{eq:snr in first lemma}
      N_1 \ge N_2 \ge m^{\max\{ p,q\}}.
      \end{align}Then  for any $\epsilon>0$, with probability at least 
    $1- 2\epsilon N^{-3}$, it holds that 
    $$ \bigg | \| \widehat A [r] \|_\F   -  \| \delta^*\|_\lt \bigg| \le   C_2 m^{-\gamma  } + C_\epsilon  \sqrt r\big\{\log^2(N) +\log^2(m) \big\}   \sqrt { \frac{m^{\max\{ p,q\}} }{N_2}}   +  4\sqrt { \sum_{k=r+1}^\infty \sigma^2_k( \delta^*)}  , $$
    where $C_2 > 0$ is an absolute constant  depending only on $C_1$, and $C_\epsilon > 0$ is an absolute constant  depending only on $C_1$ and $\epsilon$.
\end{lemma}
\begin{proof}Let  
$${\widehat \delta} [r](y, z) =\sum_{\mu=1}^{m^p} \sum_{\eta=1}^{m^q}  \widehat A[r] _{(\mu, \eta) } \Phi_\mu(y) \Phi_{\eta}(z) . $$
Observe that 
\begin{align}\label{eq:signal lower bound 1}
    \| \delta^* - \widehat \delta [r] \|_\lt \le     \| \delta^* - \delta^*_m  \|_\lt + \| \delta^*_m - \widehat \delta [r] \|_\lt = \| \delta^* - \delta^*_m  \|_\lt + \| A^* - \widehat A[r]\|_\F  ,
\end{align}
where  $\delta_m^*$ is defined in \eqref{eq:density finite matrix}, and the equality follows from  \Cref{lemma:singular value preserving}.
In addition, by \Cref{thm:matrix},
\begin{align*}
   \|  A^* - \widehat A[r]\| _\F \le 4  \sqrt { \sum_{k=r+1}^\infty \sigma^2_k(A^*)}  + 4 \sqrt r \|A^*-\widehat A \|.
\end{align*}
 Note that 
 \begin{align*}
       \sqrt { \sum_{k=r+1}^\infty \sigma^2_k(A^*)} = &\sqrt { \sum_{k=r+1}^\infty \sigma^2_k(\delta^*_m)} 
     \le   \sqrt { \sum_{k=r+1}^\infty \sigma^2_k(\delta^*_m -\delta^*)} + \sqrt { \sum_{k=r+1}^\infty \sigma^2_k( \delta^*)} \\ 
     \le &\| \delta^*_m - \delta^*\|_\F + \sqrt { \sum_{k=r+1}^\infty \sigma^2_k( \delta^*)}  ,
\end{align*}
where the   equality follows from  \Cref{lemma:singular value preserving}, the first inequality follows from \Cref{lemma:Mirsky in Hilbert space}, and the second inequality follows from  the fact that 
$$  \| \delta^*_m - \delta^*\|_\F = \sqrt { \sum_{k= 1}^\infty \sigma^2_k(\delta^*_m -\delta^*)}  \ge \sqrt { \sum_{k=r+1}^\infty \sigma^2_k(\delta^*_m -\delta^*)}. $$
Therefore 
\begin{align}
    \label{eq:signal lower bound 2} \|  A^* - \widehat A[r]\| _\F \le 4\| \delta^*_m - \delta^*\|_\F + 4\sqrt { \sum_{k=r+1}^\infty \sigma^2_k( \delta^*)} +4  \sqrt r \| A^*- \widehat A\| .
\end{align}  
Combining \eqref{eq:signal lower bound 1} and \eqref{eq:signal lower bound 2}, we have
\begin{align*}
    \| \delta^* - \widehat \delta [r] \|_\lt  &\le 5\| \delta^*_m - \delta^*\|_\F + 4\sqrt { \sum_{k=r+1}^\infty \sigma^2_k( \delta^*)} +4  \sqrt r \| A^*- \widehat A\|
    \\
    &\le C_2 m^{-\gamma }  + 4\sqrt { \sum_{k=r+1}^\infty \sigma^2_k( \delta^*)} +4  \sqrt r \| A^*- \widehat A\|,
\end{align*}   
where the last inequality follows from \eqref{eq:approximation errors}.
By \Cref{lemma:poisson matrix Bernstein} and the  fact that $A^* =\E(\widehat A)$ as specified in \eqref{eq:A hat A star}, we have with probability at least $1 - 2\epsilon N^{-3}$
\begin{align*}
\| \widehat A - \E(\widehat A) \|  
\le C_\epsilon^{\prime} \sqrt{\frac{\max\{ \| \lambda^*\|_\infty , \| \lambda^*_a\|_\infty\}(m^p + m^q)}{N_1\wedge N_2}}\log^2(N \vee (m^p + m^q)).
 \end{align*}
Therefore 
 with probability at least 
    $ 1-2 N^{-3}$, it holds that 
   \begin{align*}
  &  \bigg|     \| \widehat \delta [r]\|_\lt  -  \| \delta^*\|_\lt  \bigg| \le  \|\delta^* -\widehat \delta [r] \|_\lt 
       \\ \le&    C_2 m^{-\gamma  } + C_{\epsilon}^{\prime}\sqrt r \sqrt{\frac{\max\{ \| \lambda^*\|_\infty , \| \lambda^*_a\|_\infty\}(m^p + m^q)}{N_1\wedge N_2}}\log^2(N \vee (m^p + m^q)) +  4\sqrt { \sum_{k=r+1}^\infty \sigma^2_k( \delta^*)}
       \\
       \le &     C_2 m^{-\gamma  } + C_\epsilon \sqrt r \big\{\log^2(N) +\log^2(m) \big\}  \bigg(  \sqrt { \frac{m^{p}+m^{q}}{N_2}} \bigg) +  4\sqrt { \sum_{k=r+1}^\infty \sigma^2_k( \delta^*)},
   \end{align*} 
where the last inequality follows from  \eqref{eq:snr in first lemma}.
The desired result follows from the fact that $\| \widehat \delta   [r]\|_\lt =\| \widehat A[r]\|_\F.$
\end{proof}

\subsection{Poisson Matrix Properties}

\begin{lemma} \label{lemma:poisson matrix Bernstein}    
Let $\widehat A$ be defined in \eqref{eq:cusum poisson}. 
with probability at least $1-2N^{-3}$ 
\begin{align*}
\| \widehat A - \E(\widehat A) \|  
\le C\sqrt{\frac{\max\{ \| \lambda^*\|_\infty , \| \lambda^*_a\|_\infty\}(m^p + m^q)}{N_1\wedge N_2}}\log^2(N \vee (m^p + m^q)).
 \end{align*}

\end{lemma} 
 \begin{proof}
 Under \Cref{eq:definition appendix model}, we can apply \Cref{thm:bernstein_beta_mixing_ppp_rect_nonstat}, with $Y_i=\widehat A^\ii$, $F_{\mu,\eta}(y,z) = \Phi_{\mu}(y)\Psi_{\eta}(z)$, $1 \le \mu \le m^p$ and $1 \le \mu \le m^q$. We have with probability at least $1 - 2N^{-3}$,
 \begin{align*}
\Big\|\sum_{i=1}^{N_1} \widehat A^\ii - \E(\sum_{i=1}^{N_1} \widehat A^\ii )\Big\|
\;\le\;
C\sqrt{N_1}\big(\sqrt{\nu_{\max}} + L\big)\log^2(N \vee (m^p + m^q)),
\end{align*}
and
 \begin{align*}
\Big\|\sum_{i=N_1+1}^{N_1+N_2} \widehat A^\ii - \E(\sum_{i=N_1+1}^{N_1+N_2} \widehat A^\ii )\Big\|
\;\le\;
C\sqrt{N_2}\big(\sqrt{\nu_{\max}} + L\big)\log^2(N \vee (m^p + m^q)),
\end{align*}
where
\[
\nu_i
=\max\Big\{
\big\|\int_{\mathbb X}F(x)F(x)^\top \lambda_i^*(x)\,dx\big\|_{\op},\;
\big\|\int_{\mathbb X}F(x)^\top F(x)\lambda_i^*(x)\,dx\big\|_{\op}
\Big\},
\qquad
\nu_{\max}=\max_{1\le i\le n}\nu_i.
\]
By the orthogonality of $\Phi$ and $\Psi$, the proof of Lemma 5 in \cite{xu2025supp} gives that
$$\nu_{\max} \le \max\{ \| \lambda^*\|_\infty , \| \lambda^*_a\|_\infty\}(m^p + m^q).$$
Therefore, we have with probability at least $1-2N^{-3}$ 
\begin{align*}
\| \widehat A - \E(\widehat A) \|  &\le   \frac{1}{N_1}\|   \sum_{i=1}^{N_1} \widehat A^\ii - \E(\sum_{i=1}^{N_1} \widehat A^\ii )   \|  + 
 \frac{1}{N_2}\|  \sum_{i=N_1+1}^{N_1+N_2 } \widehat A^\ii - \E(\sum_{i=N_1+1}^{N_1+N_2 } \widehat A^\ii )   \|\\
 &\le C\sqrt{\frac{\max\{ \| \lambda^*\|_\infty , \| \lambda^*_a\|_\infty\}(m^p + m^q)}{N_1\wedge N_2}}\log^2(N \vee (m^p + m^q)).
 \end{align*}
\end{proof}

\section{Auxiliary Results}

\begin{lemma}[Lemma 25 in \citealt{xu2025supp}: Mirsky in Hilbert space] \label{lemma:Mirsky in Hilbert space} 
    Suppose $A$ and $B$ are two compact operators in $\mathcal W\otimes \mathcal W'$, where $\mathcal W$ and $\mathcal W'$ are two separable Hilbert spaces. Let 
$ \{ \sigma_k (A)   \}_{k=1}^ \infty    $  be the singular values of $\mathcal A$ in the decreasing order,  and $ \{ \sigma_k (\mathcal B)  \}_{k=1}^\infty    $   be the singular values of $\mathcal B$ in the decreasing order. Then 
$$\sum_{k=1}^\infty (\sigma_k(A) -\sigma_k(B))^2 \le \|A - B \|_\F ^2  =  \sum_{k=1}^ \infty   \sigma_{k} ^2 (A  - B)     . $$
\end{lemma}

\begin{lemma} \label{lemma:singular value preserving}
   For any  $A \in \mathbb R^{m^p\times m^q}$, let   $\mathcal F$ be a map from $\mathbb  R^{m^p\times m^q} $ to $\lt(\mathbb R^p \times \mathbb R^q)$ such that 
\begin{align} \label{eq:function finite matrix}
   \mathcal  F(A)  (y,z) = \sum_{\mu=1}^{m^p} \sum_{\eta=1}^{m^q} A_{\mu,\eta}  \Phi_{\mu} (y)\Psi_{\eta}(z),
\end{align} 
where $\{\Phi_{\mu} (y)\}_{\mu=1}^{m^p}$ and $\{\Psi_{\mu} (z)\}_{\eta=1}^{m^q}$ are generic basis functions of $ \lt(\mathbb R^p) $ and $\lt(\mathbb R^q)$ respectively. Then 
$$ \sigma_k (A) = \sigma_k(\mathcal  F(A)) \quad \text{for any } k \in \mathbb Z^+.$$

\end{lemma}
\begin{proof}
  Note that  the map $\mathcal{F}$ is distance  preserving in the sense that for any $A, B \in \mathbb R^{m^p\times m^q}$,
$$\| A-B\| _\F = \|\mathcal  F(A) -\mathcal  F(B) \|_\lt.$$
Since $\mathcal F$ is distance preserving,  $\mathcal F$ also preserves the singular values.   
\end{proof}

\begin{theorem}\label{thm:matrix}
Let $X$ and $Z$ be two generic matrices in $\mathbb R^{p\times q}$ and  that 
$$Y=X+Z .$$ Denote 
$ Y[r]= \svd(Y, r)$. Then
\begin{equation*}
\|Y[r]-X^*\|_{\mathrm{F}}
\ \le\ (2+\sqrt{2})\Big(\sqrt{ \sum_{i=r+1}^{\min\{m,n\}}\sigma_i^2(X^*)}  +\sqrt{r}\,\|Z\|\,\Big),
\end{equation*}
 
\end{theorem}
\begin{proof}
The desired result is a direct consequence of Lemma 16 in \cite{xu2025supp}.
\end{proof}

\begin{theorem}[Theorem 2.2 in \citealt{baddeley2007spatial}: Campbell's Theorem]\label{lemma:campbell}
Let $X\subset \mathbb R ^d$ be a Poisson  process   with the intensity function $\lambda^*$.  For all measurable function $f: \mathbb R ^d \to \mathbb{R}$,  it holds that
    \[
     \E ( \sum_{x\in X} f(x))= \int_{ \mathbb R^d}f(x)\lambda^*(x)dx.
    \]
\end{theorem}

\begin{lemma}
\label{lem:tail-projection}
Let $\lambda \in \lt(\Omega^{\otimes p}\times \Omega^{\otimes q})$.
For any integer $M\ge 1$, define the finite-dimensional subspaces
\[
\mathcal U_M = \mathrm{span}\{\Phi_\mu:\mu=1,\ldots,M^p\}\subset \lt(\Omega^{\otimes p}),
\qquad
\mathcal V_M = \mathrm{span}\{\Psi_\eta:\eta=1,\ldots,M^q\}\subset \lt(\Omega^{\otimes q}),
\]
and let $\mathcal P_{\mathcal U_M}$ and $\mathcal P_{\mathcal V_M}$ be the orthogonal projections onto
$\mathcal U_M$ and $\mathcal V_M$.

Let $T_\lambda:\lt(\Omega^{\otimes q})\to \lt(\Omega^{\otimes p})$ be the integral operator
\[
(T_\lambda g)(y)= \int_{\Omega^{\otimes q}} \lambda(y,z)\,g(z)\,dz .
\]
Then $T_\lambda$ is Hilbert--Schmidt (hence compact), and we define $\{\sigma_k(\lambda)\}_{k\ge1}$
as the singular values of $T_\lambda$ (in nonincreasing order).

Define the matrix $\mathcal M(\lambda)\in\mathbb R^{M^p\times M^q}$ by
\[
\mathcal M(\lambda)_{(\mu,\eta)}
= \iint_{\Omega^{\otimes p}\times \Omega^{\otimes q}}
\lambda(y,z)\,\Phi_\mu(y)\,\Psi_\eta(z)\,dy\,dz.
\]
Let $\{\sigma_k(\mathcal M(\lambda))\}_{k\ge1}$ denote the singular values of the matrix
$\mathcal M(\lambda)$.
Then for every integer $R\ge1$,
\[
\sum_{k>R}\sigma_k^2\!\big(\mathcal M(\lambda)\big)
\;\le\;
\sum_{k>R}\sigma_k^2(\lambda).
\]
\end{lemma}

\begin{proof}
\textbf{Step 1: $\mathcal M(\lambda)$ represents the compressed operator.}
Let
\[
\Pi_M = \mathcal P_{\mathcal U_M}\,T_\lambda\,\mathcal P_{\mathcal V_M}
: \lt(\Omega^{\otimes q}) \to \lt(\Omega^{\otimes p}).
\]
For $1\le \mu\le M^p$ and $1\le \eta\le M^q$,
using $\mathcal P_{\mathcal U_M}\Phi_\mu=\Phi_\mu$ and $\mathcal P_{\mathcal V_M}\Psi_\eta=\Psi_\eta$,
\begin{align*}
\langle \Phi_\mu,\, \Pi_M \Psi_\eta\rangle_{\lt(\Omega^{\otimes p})}
&= \langle \Phi_\mu,\, T_\lambda \Psi_\eta\rangle \\
&= \int_{\Omega^{\otimes p}} \Phi_\mu(y)\left(\int_{\Omega^{\otimes q}}\lambda(y,z)\Psi_\eta(z)\,dz\right)dy \\
&= \iint_{\Omega^{\otimes p}\times \Omega^{\otimes q}}\lambda(y,z)\Phi_\mu(y)\Psi_\eta(z)\,dy\,dz \\
&= \mathcal M(\lambda)_{(\mu,\eta)}.
\end{align*}
Hence, the restriction $\Pi_M|_{\mathcal V_M}:\mathcal V_M\to \mathcal U_M$ has matrix
$\mathcal M(\lambda)$ in the orthonormal bases $\{\Psi_\eta\}$ and $\{\Phi_\mu\}$.
Therefore the nonzero singular values of $\Pi_M$ coincide with those of $\mathcal M(\lambda)$,
and
\[
\sum_{k>R}\sigma_k^2\!\big(\mathcal M(\lambda)\big)
=
\inf_{\mathrm{rank}(A)\le R}\|\Pi_M-A\|_{\mathrm{HS}}^2,
\]
where $\|\cdot\|_{\mathrm{HS}}$ is the Hilbert--Schmidt norm.

\textbf{Step 2: Eckart--Young for Hilbert--Schmidt operators.}
Since $T_\lambda$ is Hilbert--Schmidt, it admits an SVD
$T_\lambda=\sum_{k\ge1}\sigma_k(\lambda)\,u_k\otimes v_k$ with orthonormal $\{u_k\},\{v_k\}$.
Let $T_{\lambda,R} =\sum_{k=1}^R\sigma_k(\lambda)\,u_k\otimes v_k$.
Then $T_{\lambda,R}$ has rank at most $R$ and
\[
\inf_{\mathrm{rank}(B)\le R}\|T_\lambda-B\|_{\mathrm{HS}}^2
=\|T_\lambda-T_{\lambda,R}\|_{\mathrm{HS}}^2
=\sum_{k>R}\sigma_k^2(\lambda).
\]

\textbf{Step 3: Compression by orthogonal projections is a contraction in HS norm.}
For any Hilbert--Schmidt operator $A$ and orthogonal projections $P,Q$,
\[
\|PAQ\|_{\mathrm{HS}}^2
= \mathrm{tr}\big((PAQ)^*(PAQ)\big)
= \mathrm{tr}\big(QA^*PAQ\big)
\le \mathrm{tr}\big(QA^*AQ\big)
\le \mathrm{tr}(A^*A)
= \|A\|_{\mathrm{HS}}^2,
\]
where we used $0\le P\le I$ and $0\le Q\le I$ and the monotonicity of the trace on positive
trace-class operators.

\textbf{Step 4: Compare best rank-$R$ errors.}
Let $T_{\lambda,R}$ be defined in \textbf{Step 2}. Then
$\mathcal P_{\mathcal U_M}\,T_{\lambda,R}\,\mathcal P_{\mathcal V_M}$ has rank at most $R$.
By optimality of the best rank-$R$ approximation of $\Pi_M$ and \textbf{Step 3},
\begin{align*}
\sum_{k>R}\sigma_k^2\!\big(\mathcal M(\lambda)\big)
&= \inf_{\mathrm{rank}(A)\le R}\|\Pi_M-A\|_{\mathrm{HS}}^2 \\
&\le \|\Pi_M-\mathcal P_{\mathcal U_M}T_{\lambda,R}\mathcal P_{\mathcal V_M}\|_{\mathrm{HS}}^2 \\
&= \|\mathcal P_{\mathcal U_M}(T_\lambda-T_{\lambda,R})\mathcal P_{\mathcal V_M}\|_{\mathrm{HS}}^2 \\
&\le \|T_\lambda-T_{\lambda,R}\|_{\mathrm{HS}}^2
= \sum_{k>R}\sigma_k^2(\lambda).
\end{align*}
This proves the claimed inequality.
\end{proof}

\begin{lemma}\label{lemma:eigen decay}
    Suppose $\{ \sigma_k(F)\}_{k=1}^\infty$ are the singular values the function $F$ and that $\{ \sigma_k(F)\}_{k=1}^\infty$ decays at a polynomial rate. Then there exists a constant $r$ depending only on the decay rate of $\{ \sigma_k(F)\}_{k=1}^\infty$ such that 
    $$ 
\sqrt{ \sum_{k=r+1}^\infty \sigma_k^2 (F) }  \;\le\; \frac{\| F \|_\lt}{5}.
  $$
\end{lemma}
    \begin{proof}
        Without lost of generality, suppose  that
        $ \sigma_k(F) = c k^{-a}$ and $\|F\|_\lt =1$. Since   $$ \sum_{k=r+1}^\infty  \sigma_k^2 (F ) = \sum_{k=1}^\infty c^2 k^{-2a}  =\|F\|_\lt^2 =1  ,$$  it follows that $a>1/2$. Since $a>1/2$, 
        $$\lim_{r\to \infty }\sum_{k=r+1}^\infty c^2  k^{-2a }=0 .$$
        Therefore there exists $r\in \mathbb Z^+$ depending only on $a$ and $c$ such that 
        $\sqrt{ \sum_{k=r+1}^\infty \sigma_k^2 (F) }  \;\le\; \frac{\| F \|_\lt}{5}$
    \end{proof}

 \begin{theorem}[Lemma 5 in \citealt{xu2025supp}: Matrix Bernstein inequality for a Poisson point process]
\label{thm:bernstein_ppp}
Let $X$ be an inhomogeneous Poisson point process on a compact set
$\mathbb X\subset\mathbb R^d$ with intensity $\lambda:\mathbb X\to\mathbb R_+$
satisfying $\|\lambda\|_\infty<\infty$.
Let $F:\mathbb X\to\mathbb R^{d_1\times d_2}$ be measurable and assume
$\sup_{x\in\mathbb X}\|F(x)\|_{\op}\le L<\infty$.
Define the matrix variance proxy
\[
\nu
=
\max\!\left\{
\left\|\int_{\mathbb X}F(x)F(x)^\top\lambda(x)\,dx\right\|_{\op},
\ \left\|\int_{\mathbb X}F(x)^\top F(x)\lambda(x)\,dx\right\|_{\op}
\right\}.
\]
Then for all $t\ge 0$,
\[
\mathbb P\!\left(
\left\|\sum_{x\in X}F(x)\;-\;\int_{\mathbb X}F(x)\lambda(x)\,dx\right\|_{\op}
\ge t
\right)
\le (d_1+d_2)\exp\!\left(-\frac{t^2/2}{\nu+Lt/3}\right).
\]
Equivalently, for any $a\ge 2$, with probability at least
$1-2(\max\{d_1,d_2\})^{\,1-a}$,
\[
\left\|\sum_{x\in N}F(x)\;-\;\int_{\mathbb X}F(x)\lambda(x)\,dx\right\|_{\op}
\le
\sqrt{2a\,\nu\log(d_1+d_2)}+\frac{2a}{3}L\log(d_1+d_2).
\]
\end{theorem}

\begin{theorem}[Theorem 1 in \citealt{banna2016bernstein}: Matrix Bernstein inequality for geometrically $\beta$-mixing self-adjoint matrices]
\label{thm:bernstein_beta_mixing_sa}
Let $\{X_i\}_{i=1}^n$ be a sequence of centered, self-adjoint random matrices in
$\mathbb R^{d\times d}$ with $\mathbb E[X_i]=0$ and satisfying
\[
\lambda_{\max}(X_i)\le M\qquad\text{a.s. for all }i=1,\dots,n.
\]
Let $\beta(k)$ be the $\beta$-mixing coefficients
of the sequence $\{X_i\}$, and assume geometric $\beta$-mixing: there exists
$c>0$ such that
\[
\beta(k)\le e^{-c(k-1)}\qquad\text{for all }k\ge 1.
\]
Define
\[
v^2
=
\sup_{\emptyset\neq K\subset\{1,\dots,n\}}
\frac{1}{|K|}
\,\lambda_{\max}\!\left(
\mathbb E\Big(\sum_{i\in K}X_i\Big)^2
\right),
\]
and
\[
\iota(c,n)
=
\frac{\log n}{\log 2}\max\!\left\{2,\frac{32\log n}{c\log 2}\right\}.
\]
Then for all $x>0$ and all integers $n>2$,
\[
\mathbb P\!\left(
\lambda_{\max}\!\left(\sum_{i=1}^n X_i\right)\ge x
\right)
\le
d\exp\!\left(
-\frac{C x^2}{v^2 n + c^{-1}M^2 + xM\,\iota(c,n)}
\right),
\]
where $C>0$ is a universal constant.
\end{theorem}

\begin{theorem}[Matrix Bernstein for geometrically $\beta$-mixing Poisson point process windows]
\label{thm:bernstein_beta_mixing_ppp_rect_nonstat}
Let $\{X^{(i)}\}_{i=1}^n$ be a sequence of (possibly nonstationary) Poisson point processes on a compact set
$\mathbb X\subset\mathbb R^d$ with intensity functions $\lambda_i:\mathbb X\to\mathbb R_+$.
Let $F:\mathbb X\to\mathbb R^{d_1\times d_2}$ be measurable and continuous, and assume
\[
\sup_{x\in\mathbb X}\|F(x)\|_{\op}\le L<\infty .
\]
For each $i$, define the centered rectangular matrix
\[
Y_i
= \sum_{x\in X^{(i)}}F(x)\;-\;\int_{\mathbb X}F(x)\lambda_i(x)\,dx
\in\mathbb R^{d_1\times d_2},
\]
and define the (rectangular) variance proxy
\[
\nu_i
=\max\Big\{
\big\|\int_{\mathbb X}F(x)F(x)^\top \lambda_i(x)\,dx\big\|_{\op},\;
\big\|\int_{\mathbb X}F(x)^\top F(x)\lambda_i(x)\,dx\big\|_{\op}
\Big\},
\qquad
\nu_{\max}=\max_{1\le i\le n}\nu_i.
\]
Assume the dependence across windows satisfies the $\beta$-mixing condition: the sequence $\{N^{(i)}\}_{i\ge1}$ is $\beta$-mixing with
coefficients $\{\beta(k)\}_{k\ge1}$ such that for some constants $c>0$,
\[
\beta(k)\le e^{-c(k-1)},\qquad k\ge 1.
\]
Then, with probability at least $1-\delta$,
\begin{align}
\Big\|\sum_{i=1}^n Y_i\Big\|_{\op}
\;\le\;
C\sqrt{n}\big(\sqrt{\nu_{\max}} + L\big)\log^2\Big(\frac{n(d_1+d_2)}{\delta}\Big),
\end{align}
where $C>0$ is a universal constant.
\end{theorem}
\begin{proof}[Proof of \Cref{thm:bernstein_beta_mixing_ppp_rect_nonstat}]
\noindent\textbf{Step 1.}
For each $i$, define the Hermitian dilation
\[
\overline{Y}_i=\begin{pmatrix}0&Y_i\\ Y_i^\top&0\end{pmatrix}.
\]
Then $\overline{Y}_i$ is self-adjoint and
\[
\|\overline{Y}_i\|_{\op}=\|Y_i\|_{\op},
\qquad
\Big\|\sum_{i=1}^n\overline{Y}_i\Big\|_{\op}
=\Big\|\overline{\sum_{i=1}^n Y_i}\Big\|_{\op} = \Big\|\sum_{i=1}^n Y_i\Big\|_{\op}.
\]
Hence it suffices to control $\|\sum_{i=1}^n\overline{Y}_i\big\|_{\op}$. We consider the following decomposition
\[
\overline{Y}_i
=
Z_i + \overline{Y}_i\,\mathbf{1}\{\|\overline{Y}_i\|_{\op}>U\} = Z_i^0 + \mathbb{E}Z_i + \overline{Y}_i\,\mathbf{1}\{\|\overline{Y}_i\|_{\op}>U\},
\]
where, for a fixed $U>0$,
\(
Z_i=\overline{Y}_i\,\mathbf{1}\{\|\overline{Y}_i\|_{\op}\le U\}
\)
and
\(
Z_i^0=Z_i-\mathbb{E}Z_i.
\)

\smallskip
\noindent\textbf{Step 2.}
For all $i$,
\[
\lambda_{\max}(Z_i^0)\le \|Z_i^0\|_{\op}\le \|Z_i\|_{\op}+\|\mathbb{E}Z_i\|_{\op}\le U+U=2U
\quad\text{a.s.}
\]
Also, $\mathbb{E}Z_i^0=0$ by construction.
Moreover, each $Z_i^0$ is a measurable function of $N^{(i)}$ only.
By \Cref{def:model_ppp_beta_mixing}, for any $k\ge 1$,
\[
\beta_{Z_i^0}(k)
\le \beta(k) \le e^{-c(k-1)}.
\]

We apply \Cref{thm:bernstein_beta_mixing_sa} to $\{Z_i^0\}_{i \in \mathbb{Z}}$ with $M=2U$. Note that
\[
\Big\|\sum_{i=1}^n Z_i^0\Big\|_{\op}
=
\max\Big\{
\lambda_{\max}\big(\sum_{i=1}^n Z_i^0\big),\;
\lambda_{\max}\big(-\sum_{i=1}^n Z_i^0\big)
\Big\}.
\]
Applying the same bound to $(-Z_i^0)$ and a union bound yields
\[
\mathbb{P}\Big(\Big\|\sum_{i=1}^n Z_i^0\Big\|_{\op}\ge x\Big)
\le
2(d_1+d_2)\exp\Big(
-\frac{C x^2}{v_U^2 n + 4c^{-1}U^2 + 2Ux\iota(c,n)}
\Big).
\]
Equivalently, we have
with probability at least $1-\delta/2$,
\[
\Big\|\sum_{i=1}^n Z_i^0\Big\|_{\op}
\le
C_0\left(
\sqrt{\Big(n v_U^2 + 4c^{-1}U^2\Big)\log\frac{4(d_1+d_2)}{\delta}}
\;+\;
2U\iota(c,n)\log\frac{4(d_1+d_2)}{\delta}
\right).
\]

\smallskip
\noindent\textbf{Step 3.}
Summing over $i$ gives
\[
\sum_{i=1}^n\overline{Y}_i
=
\sum_{i=1}^n Z_i^0
\;+\;
\sum_{i=1}^n\mathbb{E}Z_i
\;+\;
\sum_{i=1}^n \overline{Y}_i\,\mathbf{1}\{\|\overline{Y}_i\|_{\op}>U\}.
\]
We denote the truncation bias by
\[
b(U) = \max_{i = 1}^n \Big\|\mathbb E\big[\overline{Y}_i\,\mathbf{1}\{\|\overline{Y}_i\|_{\op}>U\}\big]\Big\|_{\op}.
\]
Note that $\mathbb{E}\overline{Y}_i=0$ by \Cref{lemma:campbell}, since each $N^{(i)}$ is a PPP with intensity $\lambda_i$.
Hence $\mathbb{E}Z_i=-\mathbb{E}\big[\overline{Y}_i\mathbf{1}\{\|\overline{Y}_i\|_{\op}>U\}\big]$.
Therefore,
\[
\Big\|\sum_{i=1}^n\mathbb{E}Z_i\Big\|_{\op}
\le
\sum_{i=1}^n
\Big\|\mathbb{E}\big[\overline{Y}_i\mathbf{1}\{\|\overline{Y}_i\|_{\op}>U\}\big]\Big\|_{\op}
\le
nb(U).
\]
Define the tail event
\(
\mathcal{E}_U = \{\max_{1\le i\le n}\|\overline{Y}_i\|_{\op}\le U\}.
\)
On $\mathcal{E}_U$, the random tail term vanishes:
\(
\sum_{i=1}^n \overline{Y}_i\mathbf{1}\{\|\overline{Y}_i\|_{\op} > U \}=0.
\)
Hence on $\mathcal{E}_U$,
\[
\Big\|\sum_{i=1}^n\overline{Y}_i\Big\|_{\op}
\le
\Big\|\sum_{i=1}^n Z_i^0\Big\|_{\op}
+
nb(U).
\]

\smallskip
\noindent\textbf{Step 4.} Recall that 
\[
Y_i
= \sum_{x\in X^{(i)}} F(x) \;-\;\int_{\mathbb{X}} F(x)\lambda_i(x)\,dx
\in\mathbb{R}^{d_1\times d_2}.
\]
By \Cref{thm:bernstein_ppp}, we have that
\[
\mathbb P\big(\|\bar Y_i\|_{\op}\ge t\big)
=
\mathbb P\big(\|Y_i\|_{\op}\ge t\big)
\le
(d_1+d_2)\exp\Big(-\frac{t^2}{2(\nu_i+Lt/3)}\Big)
\le
(d_1+d_2)\exp\Big(-\frac{t^2}{2(\nu_{\max}+Lt/3)}\Big).
\]
For $\mathcal E_U$, apply a union bound:
\[
\mathbb P(\mathcal E_U^c)
=
\mathbb P\Big(\max_{1\le i\le n}\|\bar Y_i\|_{\op}>U\Big)
\le
\sum_{i=1}^n \mathbb P(\|\bar Y_i\|_{\op}>U)
\le
n(d_1+d_2)\exp\Big(-\frac{U^2}{2(\nu_{\max}+LU/3)}\Big).
\]
Set $U_* = \sqrt{2\nu_{\max}u} + 2Lu/3$ and $u = \log\big(2n(d_1+d_2)/\delta\big)$, we have
\begin{equation}\label{eq:U_*1}
\frac{U_*^2}{2(\nu_{\max}+LU_*/3)}\ge u,
\end{equation}
so
\begin{equation}\label{eq:U_*2}
\mathbb P(\mathcal E_{U_*}^c)
\le
n(d_1+d_2)\exp(-u) = \frac{\delta}{2}.
\end{equation}

\smallskip
\noindent\textbf{Step 5.} Recall that
\[
v_{U_*}^2
= \sup_{\varnothing\neq K\subset\{1,\dots,n\}}
\frac{1}{|K|}\,
\lambda_{\max}\Big(\mathbb{E}\big(\sum_{i\in K}Z_i^0\big)^2\Big).
\]
Fix a nonempty set $K\subset\{1,\dots,n\}$. Expand
\[
\mathbb{E}\Big(\sum_{i\in K}Z_i^0\Big)^2
=
\sum_{i\in K}\mathbb{E}(Z_i^0)^2
\;+\;
\sum_{\substack{i,j\in K\\ i\neq j}}
\mathbb{E}(Z_i^0Z_j^0).
\]
Thus,
\[
\lambda_{\max}\Big(\mathbb{E}\big(\sum_{i\in K}Z_i^0\big)^2\Big)
\le
\sum_{i\in K}\big\|\mathbb{E}(Z_i^0)^2\big\|_{\op}
+
\sum_{\substack{i,j\in K\\ i\neq j}}\big\|\mathbb{E}(Z_i^0Z_j^0)\big\|_{\op}.
\]
\emph{(i) Diagonal terms.}
Since $Z_i^0$ is self-adjoint, we have
\[
\mathbb{E}(Z_i^0)^2 = \mathbb{E}Z_i^2-(\mathbb{E}Z_i)^2
\preceq \mathbb{E}(Z_i)^2,
\]
and thus
\[
\big\|\mathbb{E}(Z_i^0)^2\big\|_{\op}\le \big\|\mathbb{E}Z_i^2\big\|_{\op} \le \big\|\mathbb{E}\overline{Y}_i^2\big\|_{\op} = \nu_i \le \nu_{\max}.
\]
\emph{(ii) Off-diagonal terms.}
For any $j>i$,
by Berbee's coupling lemma for $\beta$-mixing process,
there exists $(Z_j^0)'$ such that $(Z_j^0)'\stackrel{d}{=}Z_j^0$, $(Z_j^0)'$ is independent of $Z_i^0$,
and $\mathbb{P}((Z_j^0)'\neq Z_j^0)\le \beta(j-i)$.
Since $\|Z_i^0\|_{\op}\le 2U_*$ a.s., we have
\[
\big\|\mathbb{E}(Z_i^0Z_j^0)\big\|_{\op}
=
\big\|\mathbb{E}\big[Z_i^0(Z_j^0-(Z_j^0)')\big]\big\|_{\op}
\le
\mathbb{E}\big[\|Z_i^0\|_{\op}\,\|Z_j^0-(Z_j^0)'\|_{\op}\big]
\le
8U_*^2\,\beta(j-i),
\]
where the first inequality follows from Jensen's inequality.
Therefore,
\[
\sum_{i\neq j}\big\|\mathbb{E}(Z_i^0Z_j^0)\big\|_{\op}
\le
2\sum_{i<j} 8U_*^2\beta(j-i)
=
16U_*^2\sum_{h\ge 1} \#\{(i,j)\in K^2: j-i=h\}\,\beta(h)
\le
16U_*^2\,|K|\sum_{h\ge 1}\beta(h).
\]
Combining \emph{(i)} and \emph{(ii)},
\[
\lambda_{\max}\Big(\mathbb{E}\big(\sum_{i\in K}X_i\big)^2\Big)
\le
|K|\nu_{\max} + 16U_*^2|K| \sum_{h\ge 1}\beta(h) \le |K|\nu_{\max} + \frac{16U_*^2|K|}{1 - e^{-c}},
\]
where the last inequality follows from \Cref{def:model_ppp_beta_mixing}.
Divide by $|K|$ and take the supremum over $K$:
\[
v_{U_*}^2 \le \nu_{\max} + \frac{16U_*^2}{1 - e^{-c}}.
\]

\smallskip
\noindent\textbf{Step 6.} Recall that 
\begin{align*}
b(U_*)
&=
\max_{1\le i\le n}
\Big\|\mathbb E\big[\bar Y_i\,\mathbf 1\{\|\bar Y_i\|_{\op}>U_*\}\big]\Big\|_{\op} =
\max_{1\le i\le n}\int_{U_*}^\infty \mathbb P(\|\bar Y_i\|_{\op}\ge t)\,dt\\
&\le
(d_1+d_2)\int_{U_*}^\infty \exp\Big(-\frac{t^2}{2(\nu_{\max}+Lt/3)}\Big)\,dt.
\end{align*}
Before computing the above integral, we simplify the notation by letting $a =\nu_{\max} > 0$, $b=L/3 > 0$ and $\psi(t)= t^2/(2(a+bt))$. 
A direct computation gives
\[
\psi'(t)=\frac{t(2a+bt)}{2(a+bt)^2},
\qquad
\psi''(t)=\frac{a^2}{(a+bt)^3}>0,
\]
so $\psi$ is convex and $\psi'$ is nondecreasing.
By convexity, for all $t\ge U_*$,
\(
\psi(t)\ge \psi(U_*)+\psi'(U_*)(t-U_*).
\)
Therefore,
\[
\int_{U_*}^\infty e^{-\psi(t)}dt
\le
e^{-\psi(U_*)}\int_0^\infty e^{-\psi'(U_*)s}\,ds
=
\frac{e^{-\psi(U_*)}}{\psi'(U_*)}.
\]
Note that
\[
\frac{1}{\psi'(U_*)}
=
\frac{2(a+bU_*)^2}{U_*(2a+bU_*)}
=
\frac{2(a+bU_*)}{U_*}\cdot \frac{a+bU_*}{2a+bU_*}
\le
\frac{2(a+bU_*)}{U_*},
\]
since $2a+bU_*\ge a+bU_*$.
We have
\begin{align*}
b(U_*) &\le (d_1+d_2)\int_{U_*}^\infty e^{-\psi(t)}dt
\le
(d_1+d_2)\frac{2(\nu_{\max}+LU_*/3)}{U_*}\exp\Big(-\frac{U_*^2}{2(\nu_{\max}+LU_*/3)}\Big)\\
&\le 2\Big(\frac{\nu_{\max}}{U_*}+\frac{L}{3}\Big)(d_1+d_2)\exp(-u)\\
&\le \Big(\frac{\nu_{\max}}{U_*}+\frac{L}{3}\Big)\frac{\delta}{n},
\end{align*}
where the third inequality follows from \eqref{eq:U_*1}, the fourth inequality follows from \eqref{eq:U_*2}.
Finally, we have
\begin{align*}
nb(U_*) \le \delta\Big(\sqrt{\frac{\nu_{\max}}{2\log(2nd/\delta)}}+\frac{L}{3}\Big),
\end{align*}

\smallskip
\noindent\textbf{Step 7.} Combining the above steps, we have with probability at least $1 - \delta$
\begin{align*}
&\Big\|\sum_{i=1}^n\overline{Y}_i\Big\|_{\op}\\
&\le
C_0\left(
\sqrt{\Big(n v_{U_*}^2 + 4c^{-1}U_*^2\Big)\log\frac{4(d_1+d_2)}{\delta}}
\;+\;
2U_*\iota(c,n)\log\frac{4(d_1+d_2)}{\delta}
\right)
\\
&\quad+
\delta\Big(\sqrt{\frac{\nu_{\max}}{2\log(2n(d_1+d_2)/\delta)}}+\frac{L}{3}\Big)\\
&\le C_1\bigg(
\sqrt{n\Big(\nu_{\max}\log(2n(d_1+d_2)/\delta) + L^2\log^2(2n(d_1+d_2)/\delta)\Big)\log\frac{4d}{\delta}}\\
&\quad\quad\quad +
\Big(\sqrt{\nu_{\max}\log(2n(d_1+d_2)/\delta)} + L\log(2n(d_1+d_2)/\delta)\Big)\iota(c,n)\log\frac{4d}{\delta}
\bigg)\\
&\le C_1\Big(\sqrt{\nu_{\max}\log(2n(d_1+d_2)/\delta)} + L\log(2n(d_1+d_2)/\delta)\Big)\bigg(
\sqrt{n\log\frac{4(d_1+d_2)}{\delta}} + \iota(c,n)\log\frac{4(d_1+d_2)}{\delta}
\bigg)\\
&\le C_2\sqrt{n}\big(\sqrt{\nu_{\max}} + L\big)\log^2\Big(\frac{n(d_1+d_2)}{\delta}\Big).
\end{align*}
\end{proof}

\section{Online 1D Poisson point processes change detection}
\label{section:1d}

For completeness, we introduce a simplified version   of \Cref{algorithm:1d} suitable for detecting change points for Poisson Process  time series data in 1D. 

Let $\{\phi_\mu \}_{\mu=1}^\infty$ be a collection of orthonormal basis functions of $\lt( \Omega)$, where $\Omega \subset \mathbb R$ is a compact interval. Let $M$ be a positive integer. 
For each process $X^{(i)} \subset\Omega $, define its intensity matrix by
\begin{equation}
\label{eq:A-i 1d}
\widehat{\mathcal V}^{(i)}\in\mathbb R^{M },\quad 
\widehat{\mathcal V}^{(i)}_{ \mu}
= \sum_{x  \in X^{(i)}} \phi_\mu(x) .
\end{equation}

\begin{theorem}
\label{theorem:main_1d}
Assume $d=1$ and that the time series   $\{X^{(i)}\}_{i\in\mathbb Z}$ on a compact domain
$\mathbb X\subset\mathbb R$ satisfies \Cref{def:model_ppp_beta_mixing}.  Let the univariate orthonormal basis
$\{\phi_\mu\}_{\mu\ge1}$ in \eqref{eq:A-i 1d} be the Legendre polynomials, and suppose the training size
$N_{\train}$ is sufficiently large.

\smallskip
\noindent\textbf{(a) No change point.}
Assume \textbf{(M0)} in \Cref{def:model_ppp_beta_mixing} holds with intensity $\lambda^*$ and
$\|\lambda^*\|_{W^{2,\gamma}}<\infty$.
If the threshold constant $\mathcal C_\alpha$ in \Cref{algorithm:1d} is chosen sufficiently large,
then with probability at least $1-\alpha$, \Cref{algorithm:1d} never raises an alarm over the entire time horizon.

\smallskip
\noindent\textbf{(b) Single change point.}
Assume \textbf{(M1)} in \Cref{def:model_ppp_beta_mixing} holds with change point $\bb\ge N_{\train}$ and intensities
$\lambda^*,\lambda_a^*$ satisfying $\|\lambda^*\|_{W^{2,\gamma }}<\infty$ and $\|\lambda_a^*\|_{W^{2, \gamma }}<\infty$.
Let $\kappa=\|\lambda^*-\lambda_a^*\|_{\lt}$ and define
\begin{equation}
\label{eq:detection delay bound main 1d}
\Delta \;=\; \Big\lceil C_{\lag}\,\big(\log(\bb)/\kappa\big)^{\,2+1/\gamma}\Big\rceil,
\end{equation}
where $C_{\lag}$ is a sufficiently large constant depending only on $\mathcal C_\alpha$.
If the window size satisfies $W\ge \Delta$, then with probability at least $1-\bb^{-3}$,
\Cref{algorithm:1d} raises an alarm within the time interval $(\bb,\bb+\Delta]$.
\end{theorem}
\begin{proof}
    The proof of \Cref{theorem:main_1d} is similar and simpler than \Cref{theorem:main}, and will be omitted for brevity. 
\end{proof}

\begin{algorithm}[tb]
   \caption{Online 1D PPP change detection}
   \label{algorithm:1d}
\begin{algorithmic}
   \STATE {\bfseries Input:} Smoothness parameter $\gamma>0$; window size $W$; threshold constant $\mathcal C_\alpha$
   
   \STATE
   \STATE {\scriptsize $\blacktriangleright$} \textbf{Initialization Stage}
   \STATE $M \leftarrow \left\lceil W^{1/(2\gamma+1)} \right\rceil$
   \FOR{$k=1$ {\bfseries to} $W$}
      \STATE $L[k] \leftarrow \sum_{i=1}^{N_{\text{train}}-W+k-1}\widehat{\mathcal V}^{(i)} {\in \mathbb{R}^M}$ \quad (computed via \eqref{eq:A-i 1d})
   \ENDFOR

   \STATE
   \STATE {\scriptsize $\blacktriangleright$} \textbf{Detection Stage}
   \STATE $\mathtt{ALARM} \leftarrow \textsc{False}$
   \FOR{$j = N_{\text{train}}+1, N_{\text{train}}+2, \ldots$}
      \FOR{$k=1$ {\bfseries to} $W-1$}
         \STATE $L[k] \leftarrow L[k+1]$
      \ENDFOR
      \STATE $L[W] \leftarrow L[W] + \widehat{\mathcal V}^{(j-1)}$

      \FOR{$k=1$ {\bfseries to} $W$}
         \STATE $R[k] \leftarrow \sum_{i=j-W+k}^{j}\widehat{\mathcal V}^{(i)}$
      \ENDFOR

      \FOR{$k=1$ {\bfseries to} $W$}
         \STATE $n_1 \leftarrow j-W-1+k$
         \STATE $n_2 \leftarrow W-k+1$
         \STATE $\mathcal D \leftarrow n_1^{-1}L[k] - n_2^{-1}R[k]$  
         
         \IF{$\| \mathcal D\| > \mathcal C_\alpha \Big(\dfrac{1}{n_2}\Big)^{\gamma/(2\gamma+1)} \log (j)$}
            \STATE $\mathtt{ALARM}\leftarrow\textsc{True}$
            \STATE \textbf{break}
         \ENDIF
      \ENDFOR
   \ENDFOR
\end{algorithmic}
\end{algorithm}

\section{Additional Numerical Studies}
\label{section:additional-num}

In this section, we report two additional simulation studies that supplement the main-paper experiments. The first is a $10$-dimensional Poisson point process (PPP) experiment in which the intensity functions are \emph{not} exactly low-rank, evaluating the proposed method against three baselines including a neural-network-based detector. The second is a robustness analysis with respect to the rank parameter $r$ and the coordinate partition.

\subsection{10D Intensity with Non-Low-Rank Difference}
\label{section:10d-non-lowrank}

We generate a temporally dependent 10D PPP time series on $[0,1]^{10}$, where the intensity function changes from
\[
\lambda^*(x) \;=\; z_t^+\bigl(1+\sin(\pi \textstyle\sum_{j=1}^{10} x_j)\bigr)
\quad\text{to}\quad
\lambda^*_a(x) \;=\; z_t^+\bigl(1+\cos(\pi \textstyle\sum_{j=1}^{10} x_j)\bigr)
\]
at the change point $\bb=1400$, with $x\in[0,1]^{10}$. Temporal dependence is introduced through an autoregressive intensity scale $\{z_t\}$. Note that neither $\lambda^*$ nor $\lambda^*_a$ admits an exact finite-rank representation, so the matricized intensities have an infinite (geometrically decaying) singular-value tail. We use $N_{\train}=1200$ pre-change samples, $N_{\text{total}}=1800$ samples in total, and average over $100$ Monte Carlo replications.

In addition to the \textbf{Matrix}, \textbf{MMD}, and \textbf{KIE} detectors used in the main paper, we include a fourth, neural-network-based density change-point detector adapted from~\citet{gong2022neural}, which we refer to as the \textbf{NN-CUSUM} detector. This detector trains a small permutation-invariant network online to discriminate the current window from pre-change reference windows and monitors the CUSUM of its held-out score; we use the authors' default architecture and optimizer settings, fixed once and held constant across replications. As in the main-paper experiments, all thresholds are calibrated by the block-permutation procedure described in \Cref{remark:tuning parameters}, and tuning parameters of the competing detectors follow their authors' default choices.

\Cref{tab:sim_10d_nonlowrank} reports the false-alarm rate, correct-detection rate, no-alarm rate, and the average detection delay (ADD) conditional on correct detection. \Cref{fig:simu_tradeoff_10d} shows the corresponding FAP--ADD trade-off curves obtained by sweeping the threshold parameter for each method. The proposed Matrix detector achieves a substantially shorter detection delay than all three competitors at comparable false-alarm rates, even though the underlying intensity difference is not exactly low-rank.

\begin{table}[ht]
\centering
\caption{Simulation results for the dependent 10D PPP experiment on $[0,1]^{10}$ with the change occurring at time $1400$. We use $N_{\train}=1200$, $N_{\text{total}}=1800$, and $100$ Monte Carlo replications. ADD and SD are reported conditional on correct detection.}
\label{tab:sim_10d_nonlowrank}
\begin{tabular}{lcccc}
\toprule
Metric & Matrix & MMD & KIE & NN-CUSUM \\
\midrule
False Alarm        & 5\%          & 6\%             & 4\%             & 4\% \\
Correct Detection  & 95\%         & 88\%            & 90\%            & 89\% \\
No Alarm           & 0\%          & 6\%             & 6\%             & 7\% \\
ADD (SD)           & 24.21 (5.86) & 266.51 (76.57)  & 284.33 (50.28)  & 240.10 (55.81) \\
ADD (no-alarm removed) & 24.21    & 263.73          & 282.93          & 235.39 \\
\bottomrule
\end{tabular}
\end{table}

\begin{figure}[ht]
  \centering
  \includegraphics[width=0.7\linewidth]{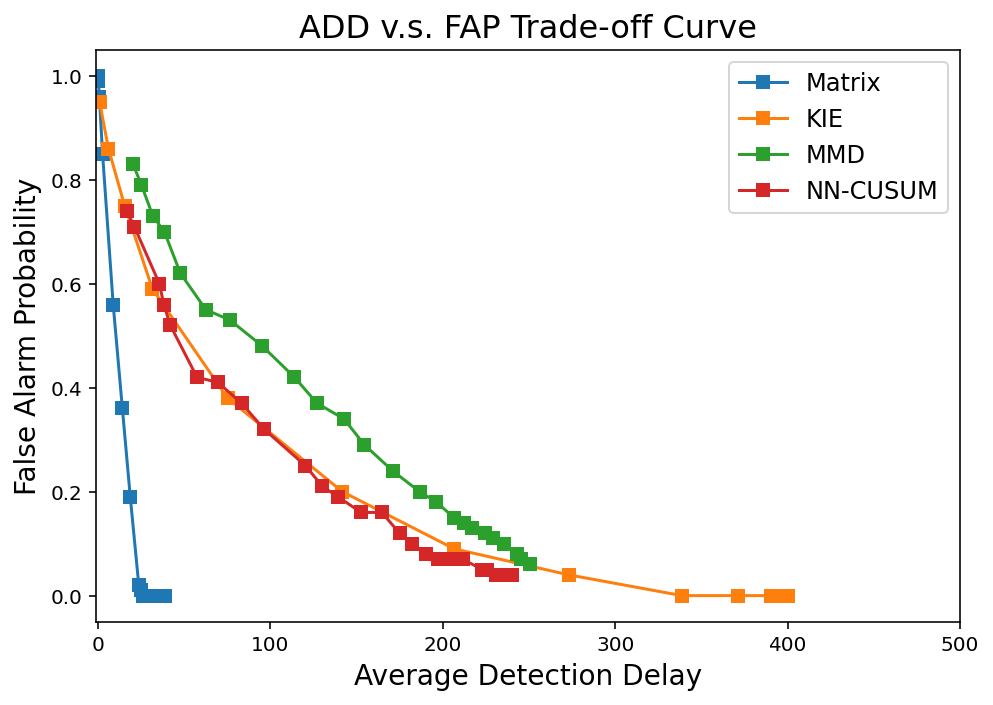}
  \caption{FAP vs.\ ADD trade-off comparison among the four detectors under the 10D non-low-rank setting.}
  \label{fig:simu_tradeoff_10d}
\end{figure}

\paragraph{Cumulative false-alarm rate under \textbf{(M0)}.}
We additionally evaluate the cumulative false-alarm rate over time on a $10$-dimensional no-change scenario generated under the same temporally dependent pre-change model as above. \Cref{tab:sim_10d_cumfar} reports the cumulative false-alarm rates measured at the time grid $\{1200,1400,\ldots,2200\}$, and \Cref{fig:cumfar_10d} shows the corresponding curves. The Matrix, MMD, and NN-CUSUM detectors all maintain a cumulative false-alarm rate within $5\%$ throughout, whereas the KIE detector's false-alarm rate drifts upward over time.

\begin{table}[ht]
\centering
\caption{Cumulative false-alarm rates for the no-change dependent 10D PPP experiment, measured at time grids $\{1200,1400,\ldots,2200\}$. The data are generated under the same temporally dependent 10D PPP pre-change model as in \Cref{tab:sim_10d_nonlowrank}.}
\label{tab:sim_10d_cumfar}
\begin{tabular}{lcccc}
\toprule
Time & MMD & KIE & Matrix & NN-CUSUM \\
\midrule
1200 & 0\%  & 0\%  & 0\%  & 0\% \\
1400 & 0\%  & 0\%  & 2\%  & 0\% \\
1600 & 1\%  & 3\%  & 2\%  & 1\% \\
1800 & 2\%  & 9\%  & 2\%  & 2\% \\
2000 & 2\%  & 11\% & 3\%  & 2\% \\
2200 & 3\%  & 14\% & 4\%  & 2\% \\
\bottomrule
\end{tabular}
\end{table}

\begin{figure}[ht]
  \centering
  \includegraphics[width=0.7\linewidth]{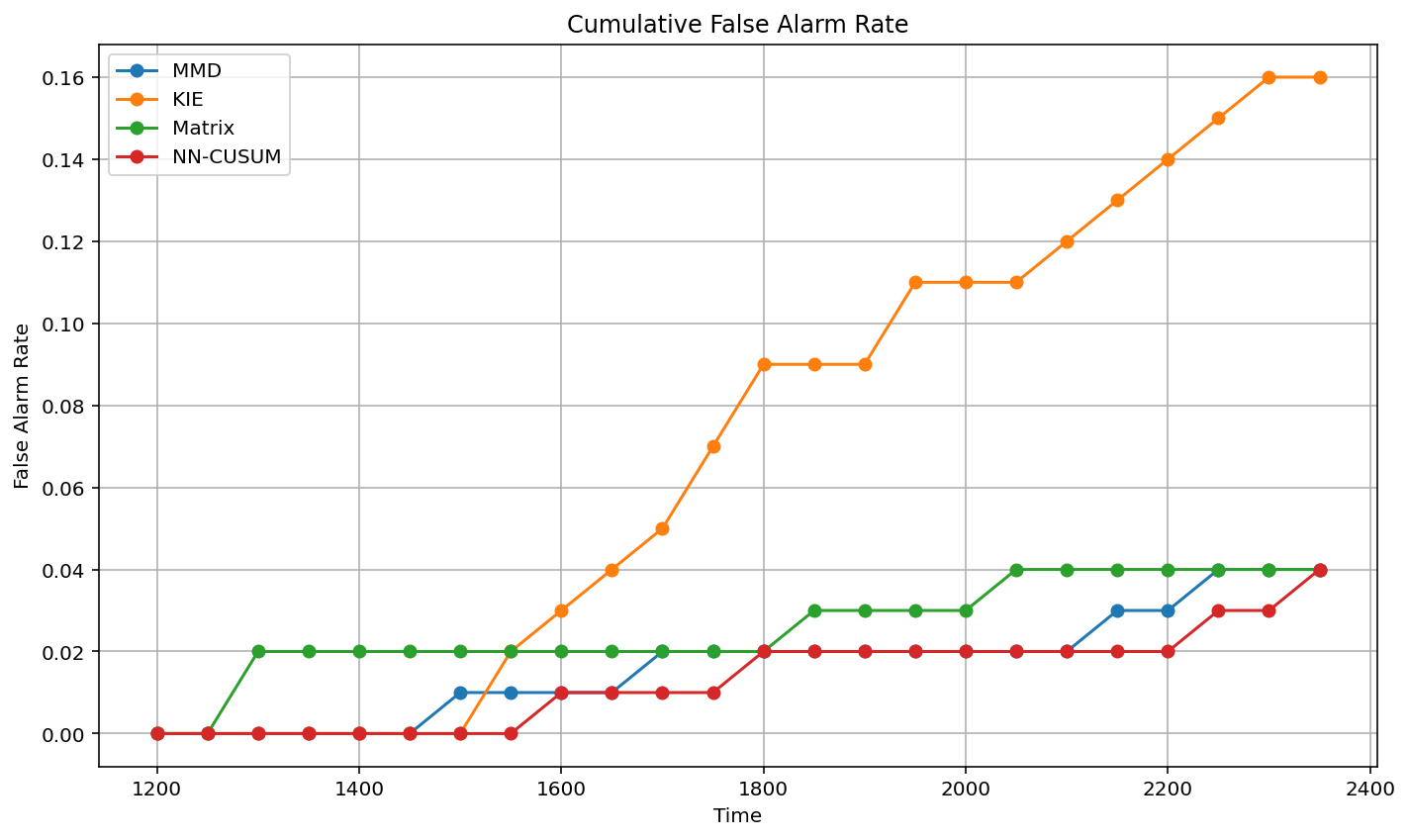}
  \caption{Cumulative false-alarm rate over time for the no-change dependent 10D PPP scenario. The Matrix, MMD, and NN-CUSUM detectors remain well within $5\%$ throughout the time horizon, while the KIE detector exhibits a noticeable upward drift.}
  \label{fig:cumfar_10d}
\end{figure}

\subsection{Robustness to Rank Selection and Coordinate Partition}
\label{section:robustness-supp}

The Matrix detector requires two structural choices: the working rank $r$ used inside the restricted SVD in \Cref{algorithm:restricted svd}, and the coordinate partition $[d]=\mathcal I_1\cup\mathcal I_2$ used to form the matricized intensity. As clarified in \Cref{remark:partition scope}, \Cref{theorem:main} provides false-alarm and detection-delay guarantees for any partition that satisfies the approximate low-rank condition, and the empirical correlation-based criterion in \Cref{remark:tuning parameters} is a practical heuristic. Here we provide an empirical assessment of how sensitive the procedure is to these two choices.

We use the same dependent 10D PPP setting as in \Cref{section:10d-non-lowrank}. We evaluate four configurations of the Matrix detector:
\begin{itemize}
\item the default correlation-based partition with rank $r\in\{3,5,10\}$, and
\item a randomly chosen coordinate partition with rank $r$ selected via the goodness-of-fit criterion of \Cref{remark:tuning parameters}.
\end{itemize}
For each configuration, we sweep the threshold and record the resulting FAP--ADD trade-off. \Cref{fig:robustness_supp} shows the four curves. The trade-off curves remain close to one another across all rank values and under the random-partition baseline, indicating that the proposed procedure is not sensitive to the specific rank choice or to the use of a non-correlation-based partition.

\begin{figure}[ht]
  \centering
  \includegraphics[width=0.7\linewidth]{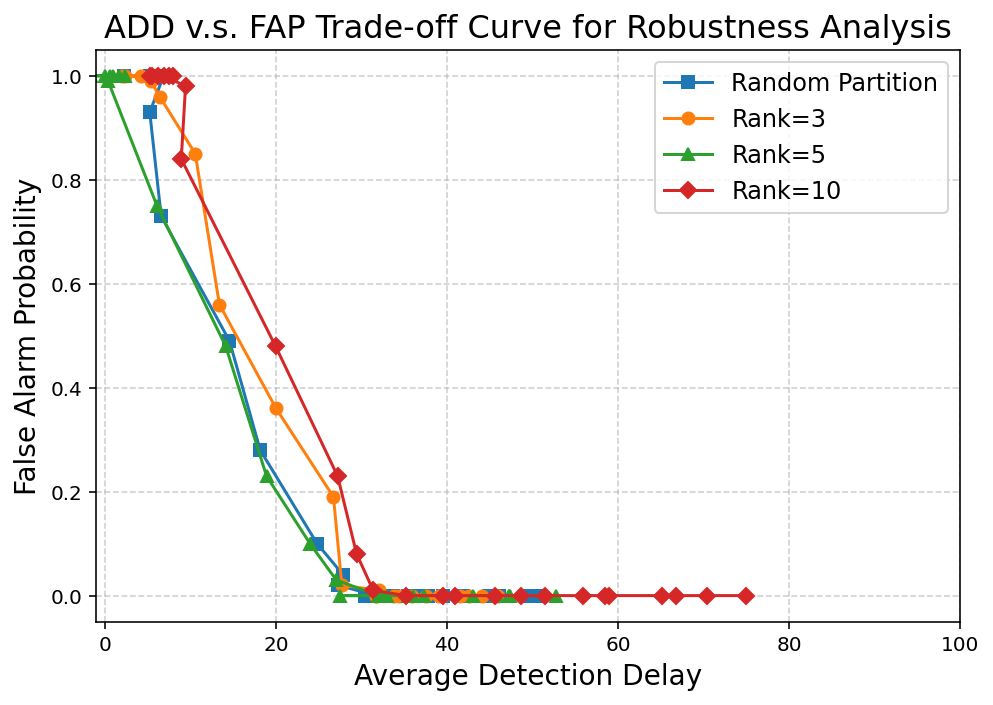}
  \caption{FAP vs.\ ADD trade-off curves for the Matrix detector under different rank choices ($r\in\{3,5,10\}$) and a random coordinate partition, in the dependent 10D PPP setting. The curves remain close to one another, indicating that the procedure is not sensitive to either choice.}
  \label{fig:robustness_supp}
\end{figure}

\end{document}